\documentclass[12pt]{article}
\pdfoutput=1
\usepackage[utf8]{inputenc}
\usepackage[margin=1in]{geometry}
\usepackage{setspace}
\usepackage{chngcntr}

\usepackage{xr}
\newcommand\rb{4} 

\usepackage{tabularx}
\usepackage{booktabs}
\usepackage{amsmath,mathtools,amsthm,amssymb,graphicx,tikz,mathtools,pgfplots,hyperref,aliascnt,subfigure,epsfig}

\newtheorem{theorem}{Theorem}
\newtheorem{lemma}{Lemma}
\newtheorem{proposition}{Proposition}

\newtheorem{example}{Example}

\newtheorem{definition}{Definition}

\usepackage{caption}
\usepackage{xfrac}
\usepackage{setspace}
\usepackage{thmtools}
\usepackage{enumitem}
\usepackage{xkeyval}
\usepackage{soul} 
\usepackage{mdframed}
\DeclareMathOperator*{\argmax}{arg\,max}
\DeclareMathOperator*{\argmin}{arg\,min}
\makeatletter
\newlength{\LP@lh@boxwidth}      
\newcommand{\LP@lh@content}{}
\newcommand{\LP@lh@tagname}{}
\renewcommand{\LP@lh@tagname}{\textbf{LP}}
{\setkeys[LP]{lhbox}{tagname,boxsep,#1}\relax%
\settowidth{\LP@lh@boxwidth}{\tagform@\LP@lh@tagname}%
\noindent%
\parbox{\LP@lh@boxwidth}{\begin{align}\tag{\LP@lh@tagname}\LP@lh@content\end{align}}%
\hspace*{\LP@lh@boxsep}%
\begin{minipage}{\linewidth-\LP@lh@boxwidth-\LP@lh@boxsep}}%
{~\end{minipage}}
\makeatother
\usepackage{soul} 
\usepackage[numbers]{natbib}
\newcommand{\param}{\theta}
\newcommand{\paramspace}{\Theta}
\newcommand{\maximal}{M}
\newcommand{\trueparam}{\theta^*}
\newcommand{\priormean}{\mu^{\circ}}
\newcommand{\posterior}[1]{\mu_{#1}}
\newcommand{\fin}{y}
\newcommand{\indicator}[1]{\mathbb{I}\{#1\}}
\newcommand{\signal}{i}

\newcommand{\discpspace}{\Theta_\delta}
\newcommand{\final}{\fin_{\pi}^*}
\newcommand{\costfxn}[2]{\beta(#2; #1)}

\newcommand{\obj}{h}
\newcommand{\costperparam}{c_1}
\newcommand{\constantcost}{c_2}
\newcommand{\I}{\mathcal{I}}

\renewcommand{\P}{\mathbb{P}}

\newcommand{\benefit}{v}
\newcommand{\ws}{\ell_S}
\newcommand{\wrr}{\ell_R}

\newcommand{\ubar}[1]{\text{\b{$#1$}}}
\newcommand{\E}{\mathbb{E}}
\newcommand{\R}{\mathbb{R}}
\newcommand{\dist}{F}
\newcommand{\popdist}{G}
\newcommand{\goal}{\mathcal{Y}}

\newcommand{\wll}[1]{\omega_{#1}^\ell}
\newcommand{\wlh}[1]{\omega_{#1}^h}
\newcommand{\lol}[1]{\ubar{\theta}_{#1}}
\newcommand{\hil}[1]{\bar{\theta}_{#1}}

\newcommand{\numintervals}{K}
 \bibpunct[, ]{(}{)}{,}{a}{}{,}%
 \def\bibfont{\small}%
\usepackage[colorinlistoftodos]{todonotes}

\usepackage{xpatch}
\usepackage{graphicx,tikz,mathtools,pgfplots,hyperref,aliascnt,subfigure}

\title{\LARGE
Information Design for Hybrid Work under Infectious Disease Transmission Risk
}

\author{Sohil Shah, Saurabh Amin, Patrick Jaillet\thanks{The authors are affiliated with the Laboratory for Information and Decision Systems, MIT, Cambridge, MA, USA, \{\texttt{sshah95},\texttt{amins}, \texttt{jaillet}\}@mit.edu}}
\date{}
\begin{document}

\maketitle
\begin{abstract}
\noindent We study a planner’s provision of information to manage workplace occupancy when strategic workers (agents) face risk of infectious disease transmission. The planner implements an information mechanism to signal information about the underlying risk of infection at the workplace. Agents update their belief over the risk parameter using this information and choose to work in-person or remotely. We address the design of the optimal signaling mechanism that best aligns the workplace occupancy with the planner's preference (i.e., maintaining safe capacity limits and operational efficiency at workplace). 

\noindent For various forms of planner preferences, we show numerical and analytical proof that interval-based information mechanisms are optimal. These mechanisms partition the continuous domain of the risk parameter into disjoint intervals and provision information based on interval-specific probability distributions over a finite set of signals. When the planner seeks to achieve an occupancy that lies in one of finitely many pre-specified ranges independent of the underlying risk, we provide an optimal mechanism that uses at most two intervals. On the other hand, when the preference on the occupancy is risk-dependent, we show that an approximately optimal interval-based mechanism can be computed efficiently. We bound the approximation loss for preferences that are expressed through a Lipschitz continuous function of both occupancy and risk parameter. We provide examples that demonstrate the improvement of proposed signaling mechanisms relative to the common benchmarks in information provision. 

\noindent Our findings suggest that information provision over the risk of disease transmission is an effective intervention for maintaining desirable occupancy levels at the workplace. Considering various preferences of the planner, our results provide the optimal signaling mechanisms for a heterogeneous workforce facing practically-driven, continuous distributions of underlying risk.
\end{abstract}

\maketitle
\section{Introduction}
\subsection{Motivation and Focus} 
The COVID-19 pandemic has generated considerable interest in the design of strategies that allow for in-person activities while controlling the risk of disease spread (\cite{nowzari_analysis_2016,drakopoulos_efficient_2014}). Although pharmaceutical tools (i.e. vaccines and medicines) are critical for reducing the impact of diseases on public health, they often need to be supplemented by other interventions. These interventions include clinical education by experts, the use of peer influencers, or informational directives from community leaders (\cite{adeagbo_improving_2022}). Such interventions become especially important after the disease becomes endemic because pharmaceuticals may lose their efficacy with time (as the genetic diversity of the disease increases) and variant-adapted pharmaceuticals might be too costly for widespread deployment~(\cite{moore_approaches_2021}).

\noindent A ``planner’’--- an entity who seeks to achieve a desired tradeoff between the  value from in-person activities and the expected cost from the resulting disease spread---has access to two broad classes of non-pharmaceutical interventions: hard and soft.  \textit{Hard interventions} are meant to control the disease spread via enforcable restrictions such as lockdowns, capacity limits, or mask mandates. When vaccines are not available or not widely accessible, such measures can be effective in flattening the contagion growth. However, relying on hard interventions in the long-term is both economically and socially costly~(\cite{birge_controlling_2020}). On the other hand, \textit{soft interventions} aim to influence agents in a  susceptible population to take actions that reduce their risk of infection -- these measures include promoting self-testing (and guidance for home isolation upon positive test) and providing information to help agents schedule their in-person activities~(\cite{ely_rotation_2021,hernandez-chanto_contagion_2021}).  Recent empirical studies have shown that public information disclosure about the risk of infection from community transmission (i.e., when the source of transmission for agents is not traceable) can be an effective tool for shaping the agents’ activity choices~(\cite{bursztyn_misinformation_2020,simonov_persuasive_2020,allcott_polarization_2020}).

\noindent In this paper, we address the following question: \emph{How should a planner disclose information over the risk of infection from community transmission in order to align the aggregate outcome of workers’ choices about in-person activity with the planner's own preferences?} We contribute to the related literature on this topic by developing approaches to Bayesian information design that account for broad range of planner preferences over the aggregate outcomes of a strategic, heterogeneous agent population in the face of a stochastic, continuously-valued risk parameter. Our results provide new insights on which classes of planner preferences have optimal information disclosure rules that necessarily exhibit a ``monotone partitional structure". We also develop a linear programming formulation that provides approximately optimal and practically implementable designs for realistic planner preferences that cannot be directly captured by stylized models~ (\cite{de_vericourt_informing_2021,hu_disclosure_2022}). 
 
\subsection{Our setting and main contributions}
We focus on a hybrid work setting. In our setting, a strategic planner provisions information to a population (workforce) of risk-neutral, heterogeneous, non-atomic agents of unit total mass. Both the planner and the agent population face the same uncertainty about the stochastic risk of infection from community transmission. We consider that this risk can be measured by a parameter (state) that is a continuous random variable with a bounded domain, with larger values of the state corresponding to a higher risk of community transmission. Each agent in the population derives value from in-person work rather than working remotely, but faces a stochastic cost associated with being infected at the workplace. This cost increases with the mass of agents at the workplace and realized value of state because of the increased frequency of close contacts and the increased risk of disease transmission per close contact, respectively. 

\noindent Intuitively, under no information about the state beyond its prior distribution, agents’ choices can lead to an outcome that overcrowds the workplace (resp. depletes in-person work) even when the true parameter is large (resp. small). Similarly, revealing full information about the true parameter can lead to agents with a high value of in-person work infecting other such agents via contact or social interactions at the workplace. More generally, under imperfect information, agents choosing in-person work must have their benefit exceed a threshold value that is increasing in the expected value of state. This motivates the basic idea behind information design -- a planner may be able to induce desirable outcomes in comparison to no- and full-information benchmarks by designing a signaling mechanism that shapes the agents’ belief about the state. 

\noindent In general, the planner’s signaling mechanism comprises of a set of signals (e.g., public health advisory, reporting of case counts, highlighting findings of latest research) and a distribution of these signals (e.g., choosing reporting service or recommendation strategy) for each possible value of the state (\cite{bursztyn_misinformation_2020, allcott_polarization_2020}). The planner uses this mechanism to signal the agents about the value of state. Agents use this signal and public knowledge of the mechanism to update their public belief about the true risk and make strategic choices on where to work (in-person or remote). If the resulting equilibrium outcome in the absence of signaling matches with the planner’s preference for all values of state, information design is unnecessary. However, in most practical settings, the planner can achieve a more desirable outcome by choosing an appropriate signaling mechanism. 

\noindent We consider that the planner’s preference is captured by a utility function over equilibrium size of in-person (or remote) population and the true state. Fundamentally, this function allows the planner to tradeoff between gain from in-person work and the cost of infections from community transmission. A number of factors contribute to this tradeoff: productivity levels of in-person/remote agents, costs of maintaining the workspace, agents’ willingness to adhere to public health guidelines (e.g., masking when sharing workspace with others) and the expected cost due to ill health and unavailability of infected agents~ (\cite{vecherin_assessment_2022},\cite{parker_return_2020}). In this paper, we assume that the planner’s preference is given and focus on the design of an optimal signaling mechanism that maximizes her expected utility in equilibrium, subject to agents’ public uncertainty about the true state and randomness in their posterior mean belief induced by the signaling mechanism. 

\noindent We now highlight the types of preferences that can be addressed by our approaches to design optimal signaling mechanisms. First, the planner may seek to maintain the size distribution of agents across in-person and remote work in a certain set that may or may not vary with the state. We consider a state-independent set-based preference where the planner specifies a single range for the size distribution of agents that is fixed across all values of the state (Sec.~\ref{sec:static_id}). 
More generally, we allow the planners utility to be a jointly Lipschitz function of the in-person equilibrium mass and the true state (Sec. \ref{sec:stateful}). Utilities of this form are general enough to cover a broad range of planner preference in practice.

\noindent By considering a continuous (and bounded) state, we extend the work of \cite{de_vericourt_informing_2021}. We believe that information provision for managing strategic agents in settings such as ours should be based on a continuous-valued state for at least two reasons. First of all, for a new disease, it may be argued that both planner and agents use beliefs over the infectiousness (hence risk from transmission) that are supported over multiple values, rather than a binary state. Secondly, as public health teams and academic research on predictive models of risk indicators becomes more advanced, it might be prudent to rely on ensemble forecasts that utilize multiple models (as opposed to stand-alone models)~(\cite{cramer_evaluation_2022}). The uncertainty resulting from such probabilistic forecasting tools can be better captured by a continuous state distribution. 

\noindent By studying the structure of information design for both state-independent and state-dependent preferences, we shed light on whether or not the optimal signaling mechanism admits a monotone partitional structure (MPS). Such a mechanism is particularly relevant to settings with continuous-valued state, since such mechanisms partition the state domain into contiguous intervals and maps each interval to a single interval-specific signal. It is then sufficient to disclose the signal corresponding to the interval that has the true state; and hence the set of signals is a strictly ordered set. In Sec.~ \ref{sec:static_id},  Theorem~ \ref{thm:r123}, we show that optimal signaling mechanism for state-independent set-based preference of the planner (i.e., when her utility is an indicator function of whether the equilibrium outcome lies within a fixed range) admits a MPS, except when agent’s prior belief on the state is not too tightly concentrated to be affected by signaling. We also obtain closed-form expressions of optimal mechanism for this setting and show performance improvement relative to no-information and full disclosure benchmarks. 

\noindent In contrast, we find that optimal signaling mechanisms for state-dependent preferences do not admit a MPS in general. Without a guarantee of MPS, the structure of the exactly optimal signalling mechanism may be arbitrary and difficult to characterize in closed-form. Instead, we seek methods that can provide an ``approximately" optimal solution whose suboptimality can become arbitrarily close to zero. Our method discretizes both the prior distribution on the state and the utility function. By using a linear programming formulation, we obtain the optimal solution to the discretized problem and show that with sufficient discretization the computed solution can achieve a value arbitrarily close to the optimal signalling mechanism (Theorem ~\ref{thm:stateful}). We do this by bounding the quality of approximate solution in terms of the Lipschitz constants of the planner’s utility function. We present a numerical study to demonstrate fast convergence of the approximation error as the level of discretization becomes finer. 
Finally, we show that our computational approach is flexible enough to recover the optimal design for binary valued state given by~\cite{de_vericourt_informing_2021} and provide near-optimal designs for other types of realistic utility functions that cannot be readily handled using earlier approaches (Sec. \ref{subsec:numerical} and \ref{subsec:dv_comp}). 

\subsection{Related Work} 
Our work is related to broader area of information design in the economics community, starting from similar work of \cite{kamenica_bayesian_2011} and well-surveyed in \cite{candogan_information_2020, bergemann_information_2019, kamenica_bayesian_2019}. In recent years, there has been a considerable interest in identifying information design problems for which an optimal signaling mechanism exhibits MPS and also settings in which such a structure is not retained (\cite{dworczak_simple_2019,guo_interval_2019,candogan_optimal_2021,ivanov_optimal_2015, candogan_information_2023}). Fundamentally, the set of distributions over posterior means that can be induced by a signaling mechanism exhibits an interesting property: the extreme points of this set correspond to all possible interval-based signal mechanisms (\cite{kleiner_extreme_2021}). Naturally, mechanisms satisfying MPS also have this extremal nature. For the case of state-independent set-based preference we obtain tight conditions for the optimal signaling mechanism to admit MPS. Here we leverage the equivalence between signaling mechanisms over continuous state and mean-preserving contractions of the parameter's prior distribution (\cite{gentzkow_rothschild-stiglitz_2016}).  Using our computational approach, we can address a variety of general state-dependent preferences which are hard to tackle analytically, and demonstrate that while retaining the interval-based structure one can compute approximately optimal signaling mechanisms with time complexity polynomial in reciprocal of error. In this sense, our work is the first one to establish the practical relevance of signaling mechanisms with interval-based structure to a fairly generic class of planner preferences. 

\noindent Another line of related work pertains to the recent work on the design of soft interventions to mitigate disease spread. Examples include: optimal design of rotation schemes of safe in-person work (\cite{ely_rotation_2021}); identifying conditions when fully information disclosure by the planner maximizes expected social welfare (\cite{hernandez-chanto_contagion_2021}); and optimal disclosure strategy for maximizing welfare in a healthcare congestion game (\cite{hu_disclosure_2022}). All these works choose specific utility functions to model planner preferences. As mentioned before, our work addresses these limitations and also considers continuous-valued state (as opposed to simplistic treatment of binary-valued state in \cite{de_vericourt_informing_2021}). 
The results we present significantly extend the work of~\cite{shah_optimal_2022} which considers optimal design over continuous-valued state relevant to occupancy management under risk of disease transmission. While they demonstrate an optimal signalling mechanism for a simpler state-independent preference, we fully characterize optimal MPS mechanisms for a more general setting. We also extend their results by introducing a computational approach to design signaling mechanisms with asymptotically diminishing approximation loss for state-dependent planner preferences under continuous-valued state.


\section{Model and Problem Formulation}
\label{sec:model}

\subsection{Agents and Information Environment}
We consider a population of non-atomic, risk-neutral, Bayesian-rational agents (workers). For convenience, assume that the total mass of the population is unity. Each agent faces a choice to either work in-person at a common workplace ($\ws$) or remotely ($\wrr$). The total mass of agents who choose to work remotely is denoted as $\fin\in[0,1]$; the mass of agents at the workplace is then $(1-\fin)$. Agents choosing $\ws$ each receive a privately-known \emph{value} from in-person work but also incur an \emph{uncertain cost} from being infected at the workplace and possibly facing symptoms of the disease. We describe both these quantities next. 

\noindent In our model, any agent’s \emph{value} from in-person work, denoted $\benefit$, is random and follows a (publicly-known) distribution $\popdist$ over $\R_+$. This value includes the agent’s personal gain from working in-person, which can be due to benefits of a shared environment (e.g., work efficiency and collaboration with co-workers) net the cost of travel to the workplace. The quantile function associated with $\popdist$ is given by $\popdist^{-1}(u) \coloneqq \sup\{t: \popdist(t) \leq u\}$. For any $u\in[0,1]$, $\popdist^{-1}(u)$ is the threshold below which a randomly drawn value of in-person work would fall below with probability $u$. 

\noindent The \emph{uncertain cost} form being infected and/or symptomatic depends on two quantities: (i) an unknown risk parameter (or state), denoted by $\trueparam$, which captures both the uncertainty from being infected via community transmission and the disutility from being symptomatic; and (ii) the mass of agents at the workplace $(1-\fin)$, which determines the likelihood of contact between the susceptible and infected agents. 

\noindent Importantly, we treat $\trueparam$ as a continuous random variable with a common prior distribution $\dist$ defined on the interval $\paramspace \coloneqq [0,\maximal]$. One can interpret $\maximal$ as the worst-case risk that agents face based on the prior knowledge about the transmissivity and severity of the disease. The prior distribution $\dist$ over $\paramspace$ then reflects the population’s overall uncertainty of the risk as estimated by the epidemiological models developed by researchers and public health agencies. 

\noindent For notational ease, we express our subsequent modeling choices using the mass of agents choosing remote work $\fin$ rather than in-person mass $1-\fin$. For any $\fin$ and $\trueparam$, the utility of any agent with value $\benefit$ from choosing to work in-person is: 
\begin{align}
\label{eqn:utility_fxn}
    u_\benefit(\ws,\fin;\trueparam) =  \benefit - \costfxn{\trueparam}{\fin}, 
\end{align}
where $\costfxn{\trueparam}{\fin}$ denotes the expected cost incurred by the agent from being infected by the disease and facing its symptoms. Here, the subscript $\benefit$ can be regarded as ``type’’ of the agent. For simplicity, agents who choose to work remotely neither face the cost of infection nor receive the benefit of in-person work; hence $u_\benefit(\wrr,\fin;\trueparam) = 0$ for any $\trueparam$ and any $\fin$.  

\noindent Motivated by a simple epidemiological model of community transmission (see Appendix \ref{appendix:infectious_cost_model}), we assume that this cost is linear in the true state $\trueparam$ and decreasing in the mass of agents choosing remote work $\fin$, and can be expressed as: 
\begin{align}
    \label{eqn:infectious_cost}
   \costfxn{\trueparam}{\fin} \coloneqq \trueparam \costperparam(\fin)+\constantcost(\fin),
\end{align}
where $c_1, c_2: [0,1] \rightarrow \mathbb{R}$ are publicly known functions with following properties: \textit{(i)} $c_1(1) = c_2(1) = 0$; \textit{(ii)} $\costperparam$ is strictly decreasing, continuous and bounded above by a constant $C$; and \textit{(iii)} $\constantcost$ is weakly decreasing and continuous.

\noindent The \emph{planner} is a strategic entity who can implement a \emph{signaling mechanism} to publicly provision information about the true state $\trueparam$ to all agents of the population. The provision of information occurs as follows. First, the planner commits to and discloses a mechanism $\pi = \langle \I, \{z_\param\}_{\param \in \paramspace} \rangle$ where $\I$ is the set of signals and  $\{z_\param\}_{\param \in \paramspace}$ is a set of probability distributions with each $ z_\param $ denoting a distribution over the set $\I$. Next, the true state~$\trueparam$ is realized from the distribution $ \dist$ unbeknownst to the agents and the planner, and the corresponding probability distribution $z_{\trueparam}$ is used to disclose a signal to all the agents; that is, $\signal \in \I$ is publicly signaled with probability $z_{\trueparam}(\signal)$. Finally, agents use the received signal to symmetrically update their belief over $\trueparam$ and make simultaneous choices to either work in-person ($\ws$) or remotely ($\wrr$). 

\noindent Specifically, on receiving signal $\signal’ \in \I$, the agents update their belief over $\trueparam$ according to Bayes’ rule: 
\begin{align}
    \label{eqn:F_i}
    \dist_{\signal'}(t) &= \mathbb{P}[\trueparam \leq t|\signal=\signal'] = 
    \frac{\int_{0}^{t} z_\param(\signal') d\dist(\param)}{\int_0^\maximal z_\param(\signal') d\dist(\param)} \hspace{0.5em} ,
\end{align}
where $\dist_{\signal'}$ is the posterior distribution corresponding to the signal $\signal’$. For a signaling mechanism $\pi$ and each signal $\signal \in \I$, one can obtain the probability with which the signal is generated $q_\signal$ and the corresponding posterior mean of the state $\posterior{\signal}$ as follows:
\begin{align}
    q_\signal &\coloneqq \int_{0}^{\maximal} z_{\param}(\signal) d\dist(\param) \hspace{1em} [\text{signal probability}] \label{eqn:q_i}\\
    \posterior{\signal} &\coloneqq \frac{\int_{0}^{\maximal} \param z_{\param}(\signal) d\dist(\param) }{\int_{0}^{\maximal} z_{\param}(\signal) d\dist(\param) }\hspace{1em}  [\text{posterior mean}] \label{eqn:theta_i}
\end{align}

\noindent Since agents are risk-neutral, they only account for the posterior mean  (and do not consider higher-order statistics) in choosing their strategies. Hence, it is often convenient to consider \emph{direct mechanisms} where the planner performs the Bayesian update and shares the updated posterior mean corresponding to the realized signal with all agents. The direct mechanism corresponding to $\pi$ is denoted as $\mathcal{T}_{\pi} = \{(q_\signal, \posterior{\signal})\}_{\signal \in \I}$.

\subsection{Equilibrium characterization}

We adopt the concept of Bayes-Nash equilibrium to determine the outcome of agents’ strategic choices under the information provided by a signaling mechanism $\pi$ (or its direct counterpart $\mathcal{T}_{\pi}$). In particular, for a mechanism $\pi$ and realized signal $\signal\in\I$, we are interested in characterizing the equilibrium mass of remote agents $\final(\signal)$, resulting from all the agents simultaneously making their choices under the posterior belief $F_\signal$ over the state $\trueparam$. The following result shows that  $\final(\signal)$ can be simply expressed a function of the posterior mean $\posterior{\signal}$. 

\begin{proposition}
\label{prop:eqbm_formula_prop_1}
For any signal $i \in \I$ realized by mechanism $\pi$, the equilibrium mass of remote agents is given by: 
\[\final(i)  = m(\posterior{\signal}) \coloneqq \inf\{u \geq 0: G^{-1}(u) \geq c_1(u)\posterior{\signal}+c_2(u)\},\]
where $\posterior{\signal} $  is the posterior mean for signal $i$. Furthermore, at equilibrium, agents with private value of in-person work $\benefit$ choose $\ws$ if $v > m(\posterior{\signal}) $, and choose $\wrr $ otherwise. 
\end{proposition}

\noindent Intuitively, at equilibrium, the remote agent mass in response to the posterior mean $m(\posterior{\signal})$ can be obtained by decreasing the mass of remote agents $u$ until the benefit from in-person work given by the (monotone) quantile function $\popdist^{-1}(u)$ no longer exceeds the expected cost from being infected $\posterior     {\signal}\costperparam(u)+\constantcost(u)$. 

\noindent The benefit that the marginal agent derives from in-person work can be viewed as the ``critical type’’ $v^*(i) \coloneqq G^{-1}(\final(i))$. Agents will work in-person if and only if their benefit exceeds $v^*(i)$. This threshold-based equilibrium characterization plays a crucial role in our design of signaling mechanism because --- to influence the mass of remote workers --- the signaling mechanism equivalently needs to shape the posterior mean that is evaluated by $m(\cdot)$. 

\noindent Moreover, we establish that as the posterior mean $\posterior{\signal}$ of the true state $\trueparam$ increases, the mass of agents choosing in-person work weakly decreases in equilibrium because the expected cost from being infected strictly increases.  This property is captured by the monotonicity and continuity of $m(\cdot)$:
\begin{lemma}
$m(\cdot)$ is non-decreasing, bounded and continuous function of posterior mean.
\label{lemma:m_smooth}
\end{lemma}
 
\noindent See Appendix \ref{si:proofs_sec_ii} for the proofs of Prop.~\ref{prop:eqbm_formula_prop_1} and Lemma~\ref{lemma:m_smooth}. Together, these results immediately allow us to obtain equilibrium outcome for two benchmarks: no information and full information. For the case when agents have \emph{no information} beyond the prior mean of $\dist$, denoted $\priormean=\E_{\trueparam\sim\dist}[\trueparam]$, the equilibrium mass of remote agents is simply the constant $m(\priormean)$. On the other hand when they have \emph{full information} about the realized true state $\trueparam$, the equilibrium outcome is the random quantity $m(\trueparam)$. 

\noindent In Fig.~\ref{fig:m_fxn}  we illustrate how the equilibrium mass of remote agents and the threshold benefit needed by an agent to shift to in-person work varies with mean belief of the true state for various distributions $\popdist$ on the benefit of in-person work.
\begin{figure}[h]
\centering
\includegraphics[width=62mm]{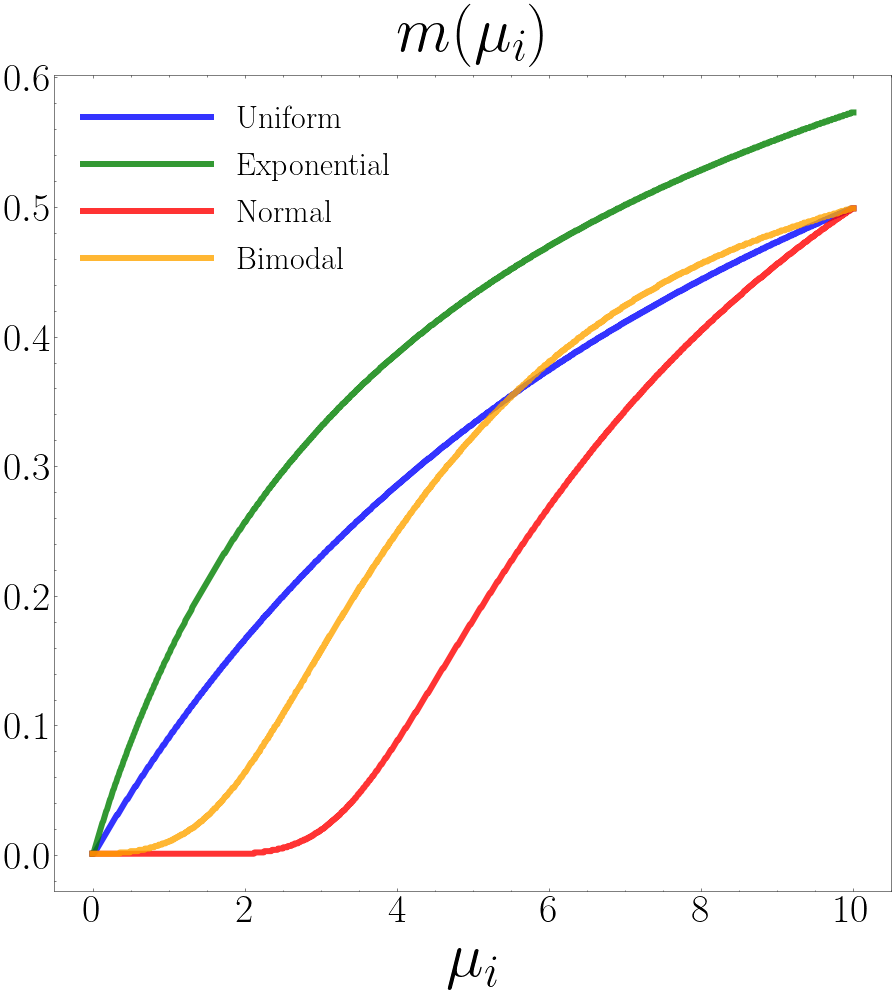}
\includegraphics[width=60mm]{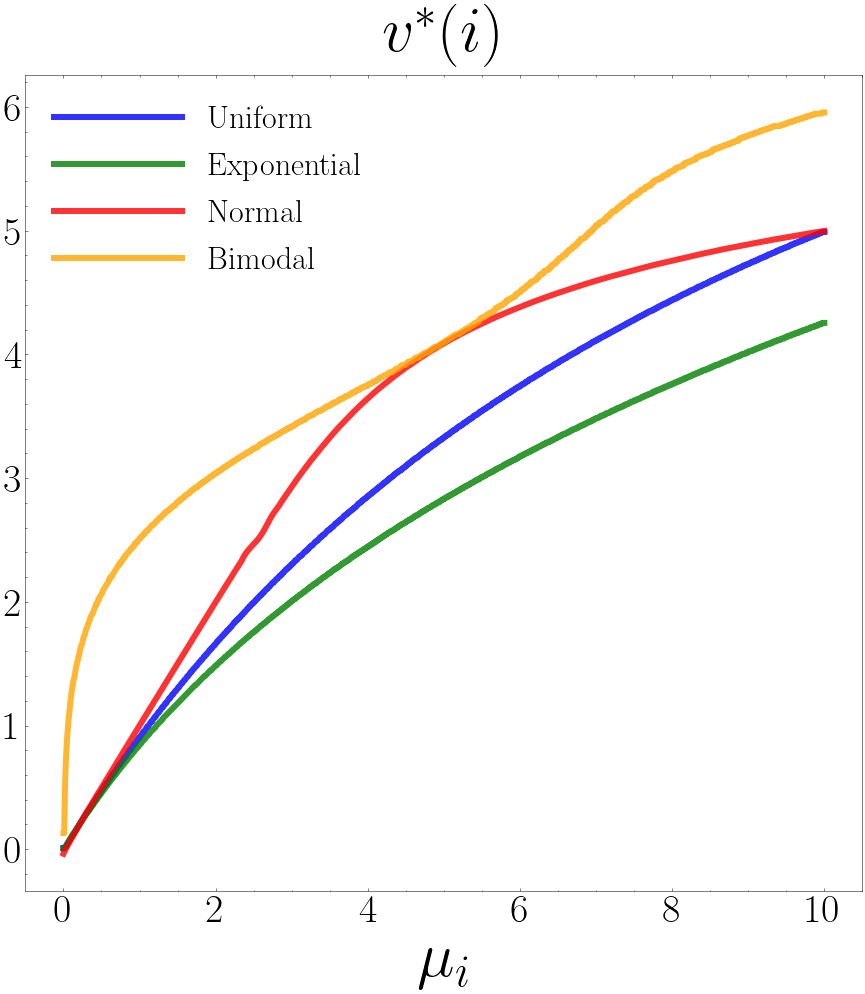}
\caption{Equilibrium mass of remote agents versus posterior mean (left) and critical type or threshold versus posterior mean (right) for various distributions $\popdist$. $c_1(u) = 1-u$, $c_2(u) = 0$.\\ Uniform: $\popdist\sim Unif[0,10]$, Exponential: $\popdist\sim Exp(\lambda = \frac{1}{5})$, Normal: $\popdist\sim \mathcal{N}(\mu=5,\sigma^2=1)$, Bimodal: $\popdist\sim \frac{1}{2}\mathcal{N}(\mu=3,\sigma^2=1)+\frac{1}{2}\mathcal{N}(\mu=7,\sigma^2=1)$.}
\label{fig:m_fxn}
\end{figure}
Here we remark that the prior distribution $\dist$ over the true state $\trueparam$ does not directly impact the equilibrium outcome -- the equilibrium outcome is only a function of the posterior belief $\posterior{\signal}$ corresponding to the signal $i$ realized by the mechanism. Also note that while $m(\cdot)$ satisfies Lemma~\ref{lemma:m_smooth}, it is not necessarily concave over the domain of state $\paramspace$.\footnote{In the literature, the concavity of $m(\cdot)$ over the region of $\paramspace$ where it assume non-zero values often plays a crucial role in the design of signaling mechanism; see for e.g., two-state setting of \cite{kamenica_bayesian_2011,bergemann_information_2019,de_vericourt_informing_2021}. In our model, this holds when $\popdist$ is a uniform or exponential distribution.}

\subsection{Planner preferences and information design problem}

In general, we consider that the planner's utility function $\obj$ is a mapping from $(\fin; \trueparam)\in[0,1]\times \paramspace$ into $\mathbb{R}_+$. For a signaling mechanism $\pi = \langle \I, \{z_\param\}_{\param \in \paramspace} \rangle $ the planner’s expected utility (objective function), denoted $V_{F,\obj}(\pi)$, is given by: 
\begin{align}
    V_{F,\obj}(\pi) = \mathbb{E}_{\trueparam \sim \dist, i \sim z_{\trueparam}}\big[\obj(\final(i); \trueparam)\big],  \label{eqn:gen_objective_stateful}
\end{align}
where $\final(i)$ is the equilibrium mass of remote agents when the signal realized by mechanism $\pi$ is $i\in\I$. 

\noindent We say that a signaling mechanism $\pi_{F,\obj}^\ast$ is optimal if it maximizes the objective function: 
\begin{align}
    \pi_{F,\obj}^\ast &\in \argmax_{\pi: \langle\I,\{z_\param\}_{\param \in \paramspace} \rangle} V_{F,\obj}(\pi).
    \label{eqn:opt_v}
\end{align}

\noindent Under the no-information (resp. full-information) environment, the value of the planner's objective is $\mathbb{E}_{\trueparam \sim \dist}\big[\obj(m(\priormean); \trueparam)\big]$ (resp. $\mathbb{E}_{\trueparam \sim \dist}\big[\obj(m(\trueparam); \trueparam)\big]$) which we show generally is not necessarily equal to the maximum achievable utility $V_{F,\obj}^\ast$ under an optimal signaling mechanism~$\pi_{F,\obj}^\ast$. Thus, we are concerned with the problem of designing~$ \pi_{F,\obj}^\ast$ that induces the agents' posterior mean beliefs on the state, with the equilibrium choices of agents resulting in an outcome that maximizes the planner's expected utility~\eqref{eqn:gen_objective_stateful}. For the sake of comparison, we will use the notation $\pi_{\mathrm{NI}}$ and $\pi_{\mathrm{FI}}$ to denote the mechanisms corresponding to no- and full-information benchmarks, respectively. 

\noindent In Sec. \ref{sec:static_id}, we first focus on the design of optimal signaling mechanisms for the setting when the planner maintains a \emph{fixed}, set-based preference over the size distribution of agents across in-person and remote work for \emph{all} values of the state -- we refer to this case as \emph{state-independent set-based preference}. Then, in Sec. \ref{sec:stateful}, we consider a general state-dependent preference model $\obj(\fin;\trueparam)$ where $\obj$ is a jointly Lipschitz function. 


\section{State-Independent, Set-Based Preferences}
\label{sec:static_id}

In this section, we consider the information design setting in which the planner maintains a set-based preference over the size distribution of agents across in-person and remote work, identical for all values of the state. This preference is represented by the union of finitely many ($\numintervals$) closed intervals: $\goal =\cup_{k=1}^{\numintervals} \Omega_k\subseteq [0,1]$ where $\Omega_k \coloneqq [\wll{k}, \wlh{k}]$. Without loss of generality, we consider that these intervals are disjoint and increasing, that is $0\leq \wlh{k} < \wll{k+1} \leq 1$ for all $k$. We refer to this setting as \emph{state-independent set-based preference}. 

\noindent The choice of $\Omega_k$ is driven by practical considerations such as desirable ranges of agent occupancies at the workplace, as driven by the number of workplace  facilities, their sizes, and minimum/maximum number of occupants and public health guidelines at each facility. As an example, $\numintervals=1$ and $\goal = [\wll{1}, \wlh{1}]$ would mean that the planner with two workplace facilities prefers the in-person mass of agents to be in one of the intervals $[0, \wll{1})$ or $(\wlh{1}, 1]$; this corresponds to an occupancy limit below $\wll{1}$ in the first facility, and a minimum (resp. maximum) limit $(\wlh{1}-\wll{1})$ (resp. ($1-\wll{1})$) in the second facility, which is to be used after the first facility's occupancy limit is reached. 

\noindent The planner's utility for state-independent set-based preferences can be defined as $\obj(\fin;\trueparam) \coloneqq \mathbb{I}\{y \in \goal\}$. Her expected objective $V_{F,\obj}(\pi)$ for a signaling mechanism $\pi = \langle \I, \{z_\param\}_{\param \in \paramspace} \rangle$ then becomes~$V_{F,\obj}(\pi) = \mathbb{P}\{\final(i) \in \goal\}=\mathbb{P}\{\final({\signal}) \in \cup_{k=1}^{\numintervals} \Omega_k\}$, where $\final({\signal})$ is the agents' equilibrium remote mass in response to signal $i\in\I$ that is realized with probability $z_\trueparam(i)$. Following Prop.~\ref{prop:eqbm_formula_prop_1}, we can  write the problem of maximizing planner's objective~\eqref{eqn:gen_objective_stateful} as follows: 
\begin{align}
        V_{F,h}^* = \max_{\pi: \langle \I,  \{z_\param\}_{\param \in \paramspace}\rangle} \sum_{k=1}^{\numintervals} \mathbb{P}\{\wll{k} \leq m(\posterior{\signal}) \leq \wll{k}\}, \label{eqn:value_c_i}
\end{align}
where $\posterior{\signal}$ is the posterior mean belief induced by the signal $i$. Lemma~\ref{lemma:m_smooth} implies that for any $k=1,\ldots,\numintervals$ the preimage of $m(\cdot)$ over $\Omega_k=[\wll{k}, \wlh{k}]$ is a closed interval $[\lol{k}, \hil{k}]\eqqcolon \bar{\Theta}_k\subseteq \paramspace$. Hence, the occurrence of the event $\{m(\posterior{\signal})\in \Omega_k\}$ is equivalent to that of the event $\{\posterior{\signal} \in \bar{\Theta}_k\}$. Furthermore, the monotonicity of $m(\cdot)$ implies that $\lol{k}$ and $\hil{k}$ are increasing in $k$. By exploiting this structure, the problem~\eqref{eqn:value_c_i} can be re-written as optimization over direct mechanisms of the form $\mathcal{T}_{\pi} = \{(q_\signal, \posterior{\signal})\}_{\signal \in \I}$:  
\begin{align}
        V_{F,h}^* &= \max_{\mathcal{T}_\pi: \{(q_i, \posterior{i})\}_{i \in \I}} \sum_{i \in \I} \sum_{k=1}^{K} q_i \mathbb{I}\{\lol{k} \leq \posterior{i} \leq \hil{k}\}, 
        \label{eqn:red_1_form}
\end{align}
where $q_i$ and $\posterior{i}$ are the signal probability and posterior mean for signal~$i$ (refer to~\eqref{eqn:q_i} and \eqref{eqn:theta_i}). 

\subsection{Regimes}
\label{subsec:state_indpt_regimes}
Now consider the no-information mechanism $\pi_{\mathrm{NI}}$ which can be constructed by choosing $\I$ as a singleton set (say $\{s\}$) and $z_\param=1$ for all $\param \in \paramspace$. The corresponding direct mechanism is $\mathcal{T}_\mathrm{NI}=\{(1, \priormean)\}$, where $\priormean$ is the mean for prior distribution $\dist$. From~\eqref{eqn:red_1_form}, note that the no-information mechanism achieves maximum planner utility of $1$ if and only if there exists an interval  $k=[\numintervals]$ for which the prior mean belief $\priormean\in\bar{\Theta}_k$. On the other hand, if $\priormean\notin\bar{\Theta}_k$ for all $k\in[\numintervals]$, then the planner achieves a utility of $0$ under $\pi_{\mathrm{NI}}$.  

\noindent Consequently, to solve~\eqref{eqn:red_1_form}, it is useful to distinguish the following qualitative different cases -- which we refer to as \emph{regimes} -- based on the position of prior mean $\priormean$ relative to the intervals $\{[\lol{k},\hil{k}]\}_{k \in [K]}$. These regimes can be defined in terms of $\cup_{k=1}^K \bar{\Theta}_k$ as follows: 
\begin{itemize}
    \item[] $(\mathrm{R1})$: $\priormean \in \cup_{k=1}^K \bar{\Theta}_k$. The prior mean $\priormean$ lies in one of the intervals $\bar{\Theta}_k$ -- and as noted above $\pi_{\mathrm{NI}}$ is optimal in this regime.     
    \item[]  $(\mathrm{R2})$: $\priormean > \sup \cup_{k=1}^K\bar{\Theta}_k$. Equivalently, $\hil{\numintervals}<\priormean\leq \maximal$, where $\maximal$ is the maximum value of the state.  
    \item[]  $(\mathrm{R3})$: $\priormean < \inf \cup_{k=1}^K\bar{\Theta}_k$. Equivalently, $0\leq\priormean < \lol{1}$.
    \item[] $(\mathrm{R4})$: $\priormean \notin \cup_{k=1}^K \bar{\Theta}_k\; \wedge \; \inf \cup_{k=1}^K\bar{\Theta}_k < \priormean < \sup \cup_{k=1}^K\bar{\Theta}_k $. That is, $\priormean$ does not lie in any interval but lies in the gap between two contiguous intervals ($\exists k'\in[\numintervals]$ such that $\priormean\in( \hil{k'}, \lol{k'+1})$). 
\end{itemize}

\noindent Furthermore, in any regime, any direct mechanism~$\mathcal{T}_\pi$ that solves~\eqref{eqn:red_1_form} is not unique in general. This follows from the fact that any signal $i\in\I$ with signal probability~$q_i$ and posterior mean~$\posterior{i}$ can be branched into two signals $i_i$, $i_2$ uniformly at random to induce symmetric posterior means~$\posterior{i_1}=\posterior{i_2}=\posterior{i}$ with probability~$\sfrac{q_i}{2}$ each, and hence such a construction achieves the same planner objective. However, the following lemma shows that the search for an optimal mechanism can be limited to the class of direct mechanisms that use a set of signals~$\I$ of size at most $|\I|=\numintervals+1$ (see proof in Appendix \ref{appendix:proofs_sec_iii}). 
\begin{lemma}
\label{lemma:num_signals_bdd}
There exists a direct mechanism~$\mathcal{T}_{\pi}^\ast=\{(q_\signal, \posterior{\signal})\}_{\signal \in [\numintervals+1]}$ that achieves optimal planner objective in~\eqref{eqn:red_1_form} and satisfies following constraints:
\[\posterior{i} \in \bar{\Theta}_i,\quad i=1,\dots,\numintervals, \quad \text{and} \quad \posterior{K+1} \notin \cup_{k=1}^K \bar{\Theta}_k.\]
\end{lemma}
\noindent That is, to solve for an optimal mechanism, we need at most one signal for each of the $\numintervals$ intervals  (i.e., $\bar{\Theta}_i,\; i=1,\dots,\numintervals$) to induce posterior mean in the desirable set, and one additional signal that induces posterior mean that does not lie in any of these sets. We henceforth use this insight to search over mechanisms that can be represented by the set of tuples~$\{(q_\signal, \posterior{\signal})\}_{\signal \in [\numintervals+1]}$, and define the (cumulative) distribution of posterior means for such a set of tuples as $H(t) = \sum_{k=1}^{K+1} q_k \mathbb{I}\{\posterior{k} \leq t\}$.

We say that posterior mean distribution $H$ corresponding to~$\mathcal T_\pi=\{(q_\signal, \posterior{\signal})\}_{\signal \in [\numintervals+1]}$ is a \emph{mean-preserving contraction} if and only if for any $y\in \paramspace$, $\int_0^y H(t)dt \leq \int_0^y F(t)dt$, with equality for $y=\maximal$ (\cite{mas-colell_microeconomic_1995}). If $H$ is a mean-preserving contraction of $\dist$ we write $H \succcurlyeq \dist$ (i.e., $H$ majorizes $\dist$). This relationship can be expressed with the following equivalent constraints (\cite{candogan_optimal_2021}): 
\begin{align}
   \int_0^y (1-H(t)) dt \geq \int_0^y (1-F(t)) dt,  \forall y\in \paramspace \quad \Leftrightarrow \quad \int_0^x H^{-1}(s) ds \geq \int_0^x F^{-1}(s) ds, \forall x\in [0,1]\label{eq:MPC_constraint1}
\end{align}
with equality at $y=\maximal$ (resp. $x=1$).

\noindent Following the seminal result from~\cite{gentzkow_rothschild-stiglitz_2016} who build on~\cite{blackwell_theory_1954}, we know that set of tuples~$\{(q_\signal, \posterior{\signal})\}_{\signal \in [\numintervals+1]}$ is implementable via a signaling mechanism~$\pi$ \emph{if and only if} the corresponding posterior mean distribution~$H\succcurlyeq \dist$, i.e., $H$ is a mean-preserving contraction of the prior distribution~$\dist$. Intuitively, implementability requires that the mechanism shifts the probability mass from the tails of the prior distribution~$\dist$ ``inward'' in a manner that preserves the mean of the distribution (both $\dist$ and $H$ have equal means).  

\noindent We note that $f(x)=\int_0^x F^{-1}(s) ds$ is convex in $x$ with $f(0)=0$. Furthermore, by definition of $H$, $\int_0^u H^{-1}(s) ds$ is a piecewise linear function in $u$ with breakpoints in the set~$\{\sum_{j}^n q_j: n\in[\numintervals+1]\}$. Hence, we obtain that to ensure $H\succcurlyeq \dist$, it suffices to enforce the constraint~\eqref{eq:MPC_constraint1} at these breakpoints. The following lemma captures this observation (see proof in Appendix \ref{appendix:proofs_sec_iii}):
\begin{lemma}
\label{lemma:reduce_MPC_constraints}
For the posterior distribution of means defined as $H(t)=\sum_{k=1}^{K+1} q_k \mathbb{I}\{\posterior{k} \leq t\}$ for all $t\in\paramspace$, enforcing that $H\succcurlyeq \dist$ is equivalent to the following constraints:
\begin{subequations}\label{eq:MPCall}
\begin{align}
    &\sum_{j=1}^{n} q_j\posterior{j} \geq \int_0^{\sum_{j=1}^{n} q_j} \dist^{-1}(s)ds,\quad \forall n \in [\numintervals] \label{eqn:MPC}\\
 &\sum_{j=1}^{\numintervals+1} q_j \posterior{j} = \priormean. \label{eqn:mean_mean}
\end{align}
\end{subequations}
\end{lemma}

\subsection{Regimes with monotone paritional structure (MPS)}
\label{subsec:state_indpt_mps}
We define a class of signaling mechanisms that plays an important role in our subsequent results:
\begin{definition}
\label{defn:mps}
We say that a signaling mechanism~$\pi$ has a \emph{monotone partitional structure} (MPS) if there exists a finite partition of the state-space~$\paramspace$, defined as $\mathcal{P} \coloneqq \{\paramspace_j\}_{j=1}^n = \{(t_{j-1},t_j]\}_{j=1}^n$  for some $n$ with $0=t_0<t_1<\dots<t_{n-1}<t_n = \maximal$, such that $\I=[n]$ and for any $\param\in\paramspace$, $z_\param(j) = \mathbb{I}\{\param \in (t_{j-1}, t_j]\}$. The corresponding direct counterpart can be written as the set of tuples~$\{(q_j, \posterior{j})\}_{j \in [n]}$, where following~\eqref{eqn:q_i} and \eqref{eqn:theta_i}, $q_j = \dist(t_j) - \dist(t_{j-1})$ and $\posterior{j} = \sfrac{\int_{t_{j-1}}^{t_j} \param d\dist(\param)}{q_j}$. 
\end{definition}

\noindent We are now in the position to characterize the optimal signaling mechanism for the regimes $\mathrm{R2}$ and $\mathrm{R3}$. (We already know that $\pi_{\mathrm{NI}}$ is optimal in regime~$\mathrm{R1}$.) The following proposition shows that the optimal signaling mechanism in these regimes has a monotone partitional structure. (This property does not necessarily hold for regime $\mathrm{R4}$, as extensively discussed in Sec. \ref{subsec:state_indpt_non_mps}. We instead identify sufficient conditions for MPS to hold at optimality for regime $\mathrm{R4}$ and provide a method to find the optimal \emph{direct} signalling mechanism.) Before proceeding, we define the increasing function for any $\param\in\paramspace$:
\begin{align}
    \bar f(\param) \coloneqq \sup\left\{x:\int_0^x F^{-1}(s)ds \leq x \param\right\},\label{eq:barffn}
\end{align}
and note that for any~$\param\in\paramspace$ the constraint $\param x\geq f(x)$ is satisfied if and only if $0\leq x\leq \bar f(\param)$.

\begin{theorem}
The optimal value of planner's objective~$V_{F,h}^\ast$ and the corresponding signaling mechanism~$\pi_{F,\obj}^\ast$ for the regimes~$\mathrm{R1}-\mathrm{R3}$ are as follows:
\begin{itemize}
    \item[] $(\mathrm{R1})$: $V_{F,h}^\ast=1$ and $\pi_{F,\obj}^\ast=\langle \{1\}, \{z_\param\}_{\param \in \paramspace}\rangle$ with $z_\param(1) = 1$ for all $\param\in\paramspace$. 
    \item[]  $(\mathrm{R2})$: $V_{F,\obj}^\ast = q_1^\ast$ where  $ q_1^\ast\coloneqq \min\left\{\bar f(\hil{\numintervals}),\frac{M-\priormean}{M-\hil{\numintervals}}\right\}$ and $\pi_{F,\obj}^\ast = \langle \{1,2\}, \{z_\param\}_{\param \in \paramspace}\rangle$ with $z_\param(1) = 1$ for $\param \leq F^{-1}(q_1^\ast)$ and $z_\param(2) = 1$ for $\param > F^{-1}(q_1^\ast)$.
  
    \item[] $(\mathrm{R3})$: $V_{F,\obj}^\ast = 1-q_2^\ast$ where $q_2^\ast \coloneqq \inf\left\{q \geq \frac{\lol{1}-\priormean}{\lol{1}}:q \leq \bar f(\lol{1}-\frac{\lol{1}-\priormean}{q})\right \}$ and $\pi_{F,\obj}^* = \langle \{1,2\}, \{z_\param\}_{\param \in \paramspace}\rangle$ with $z_\param(1) = 1$ for $\param > F^{-1}(q_2^\ast)$ and $z_\param(2) = 1$ for $\param \leq F^{-1}(q_2^\ast)$. 
\end{itemize}
Thus, $\pi_{F,\obj}^\ast$ has a monotone partitional structure in regimes~$\mathrm{R1}-\mathrm{R3}$: $t_0 = 0$ and $t_1 = M$ for $(\mathrm{R1})$; $t_0 = 0$, $t_1 = F^{-1}(q_1^*)$ and $t_2 = M$ for $(\mathrm{R2})$; and $t_0 = 0$, $t_1 = F^{-1}(q_2^*)$ and $t_2 = M$ for $(\mathrm{R3})$.  
\label{thm:r123}
\end{theorem}

 \begin{figure}[h!]
    \centering
\includegraphics[width=140mm]{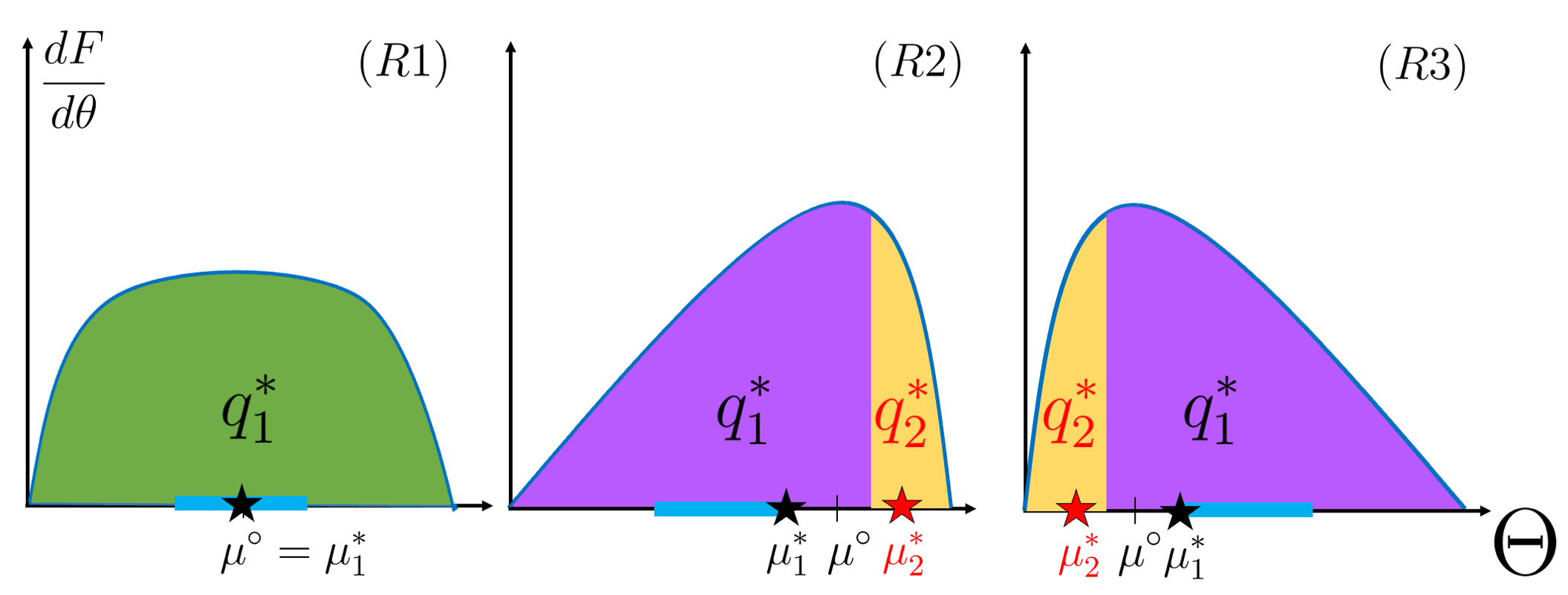}
    \caption{The probability density function for $\trueparam$, the location of prior mean~$\priormean$, and the intervals $[\lol{k},\hil{k}]$ (blue) for regimes~$(\mathrm{R1}-\mathrm{R3})$; for simplicity~$\numintervals=1$. For optimal mechanism given by~Theorem~\ref{thm:r123}, the signal probabilities are determined by the probability mass of colored regions~(green, violet, yellow) and the locations of corresponding posterior means are marked as $\star$.}
    \label{fig:r123}
\end{figure}
\noindent We provide the proof of this result in Appendix \ref{appendix:proofs_sec_iii}. The structure of optimal mechanism for each of the regimes~$(\mathrm{R1}-\mathrm{R3})$ is depicted in Fig.~\ref{fig:r123}. In $\mathrm{R1}$, the optimal mechanism is~$\pi_{\mathrm{NI}}$. It maps the entire probability mass over~$\paramspace$ to a single signal, and induced posterior mean is same as the prior mean. In regimes $\mathrm{R2}$ and $\mathrm{R3}$, the optimal mechanism partitions the  state-space~$\paramspace$ into two intervals which can understood as ``low'' and ``high'' parameter ranges, and each interval corresponds to a unique signal. The locations of posterior means induced by the two signals relative to the prior mean are also illustrated in the figure, along with the corresponding signal probabilities. These signal probabilities and posterior means define the optimal direct mechanism that is implemented by~the $\pi_{F,\obj}^\ast$ given in~Theorem~\ref{thm:r123}.

\noindent While optimal, we also demonstrate the derived optimal mechanism manifests a nontrivial practical improvement non-informative and fully-informative benchmarks. We substantiate these improvements from optimal signaling with numerical experiments described in further detail in Section \ref{subsec:numerical_state_indpt}.

\section{State-dependent Preferences}
\label{sec:stateful}

In this section, we consider state-dependent preferences; i.e., settings where the planner's preference is modeled by~$\obj(\fin;\trueparam)$  which is her utility for an equilibrium remote mass~$\fin$ and the state is $\trueparam$. In particular, we seek to solve the design problem~\eqref{eqn:opt_v} for preference models $\obj$ that are \textit{Lipschitz continuous} where the
planner's utility is Lipschitz in both~$\fin$ and~$\trueparam$.

\noindent The information design problem for these settings becomes more challenging in comparison to Sec. \ref{sec:static_id} due to the fact that the set of preferred equilibrium outcomes depends on the true state. Hence, it is no longer sufficient to characterize the optimal mechanism by analyzing the locations of induced posterior means~$\posterior{\signal}$ within the state-space~$\paramspace$. 

\noindent In fact, one can show MPS mechanisms are not necessarily optimal as the optimal signaling mechanism may ``pool" disparate intervals of the state space $\paramspace$ to the same signal (see Section \ref{ex:statefulpool} in supplementary materials). Lacking simple structural guarantees, this creates tremendous difficulty in the analytical characterization of optimal signaling for state-dependent preferences.  Hence, we adopt a computational approach to design approximately optimal signaling mechanisms, while still maintaining an interval-based structure. The approach entails discretizing the continuous distribution $\dist$ in order to limit the number of states $\param$ for which we need to consider preferences~$\obj(\cdot;\param)$ over remote agent mass. This allows us to compute an optimal solution under this discretization using a linear programming (LP) formulation. We subsequently use this LP-based design to provide approximately optimal solutions for Lipschitz continuous preference models.

\subsection{LP-based Design for Discretized Problem}

We discretize $\dist$ over the continuous space $\paramspace$ into a discrete distribution $\dist_\delta$ by taking a uniform partition of $\paramspace$ with $N \coloneqq M\delta$ intervals $[\frac{j-1}{\delta},\frac{j}{\delta})$ each of length $\frac{1}{\delta}$ and assigning all the probability mass in the interval to the minimum of that interval $\nu_j \coloneqq \frac{j-1}{\delta}$, giving us the distribution~$\dist_\delta$. 


\begin{definition}
\label{defn:discretized_dist}
The \textit{$\delta$-discretization} $\dist_\delta$ of any continuous distribution $\dist$ is a discrete probability distribution over $\nu_j$ by $\discpspace \coloneqq \{\nu_1, ...,\nu_N\}$ such that $\hat{\theta} \sim \dist_\delta$ has $\mathbb{P}[\hat{\theta} = \nu_j] \coloneqq p_{j} = \dist(\frac{j}{\delta})-F(\frac{j-1}{\delta})$.
\end{definition}

\noindent Observe that $\mathbb{E}_{\trueparam \sim \dist}[\trueparam] \geq \mathbb{E}_{\nu \sim \dist_\delta}[\nu]$ since the probability mass of each interval is shifted towards the minimum of the interval.

\noindent For a given $\fin\in[0,1]$, the function $\obj(\fin;\nu_j)$ can be evaluated for each $\nu_j\in\discpspace$. To computationally obtain an optimal design for the discretized setting, we further consider that for all $\nu_j$, the function $\obj(\cdot;\nu_j)$ is evaluated at pre-specified discrete number of points $y_0 \coloneqq 0  < y_1 < .. < y_{K-1} < y_{K}\coloneqq 1$ ($K \in \mathbb{N}$). Additionally, for all $\nu_j \in\discpspace$ and $\fin \in [y_{k-1},y_k]$, we take $\obj(\fin;\nu_j) \approx c_{jk}$, where $ c_{jk}$ for $j=1,\dots,N$ and $k=1,\dots,K$ are values corresponding to a piecewise-constant approximation of $\obj$. For convenience, we let $\mathbf{y}=\left(y_0, \dots, y_K\right)$ and $\mathbf{c}=\left(c_{jk}\right)\in\R^{N\times K}$. We are ready to state the LP-based design for the discretized setting when $\trueparam\sim \discpspace$ and $\obj(\fin;\nu)$ is piecewise-constant in $y$ (see \ref{sec:proofs_for_sec_4} for proof).  

\begin{lemma}
\label{lemma:lp_discrete}
An optimal design $\pi_{H,\obj}^\ast = \langle \I, \{z_{\param}\}_{\param \in \discpspace} \rangle$ for the discretized setting where $H$ is discrete and $h$ is piecewise-constant can be constructed from an optimal solution $\{z_{ji}^\ast\}$ of the following linear program by choosing $\I = [K]$ and, for all $\signal \in \I$ and $j \in [N]$ setting $z_{\nu_j}(\signal) = 0$ if $p_j = 0$ and $z_{\nu_j}(\signal) = \frac{z_{ji}^\ast}{p_j}$ otherwise.\\
    \begin{equation*}
\begin{array}{ll@{}ll}
\text{maximize}  & \displaystyle&\sum_{j=1}^N \sum_{i=1}^{K} c_{ji} z_{ji} &\\
\text{subject to}& \displaystyle&\sum_{i=1}^{N+1}  z_{ji} = p_j,  &j=1 ,\dots, N\\
   &\displaystyle &z_{ji} \geq 0,\hspace{0.5cm} &j=1 ,\dots, N,\quad i=1 ,\dots, K \\
    &\displaystyle m^{-1}(y_{i-1})&\sum_{j=1}^N z_{ji} \leq  \sum_{j=1}^N \nu_j z_{ji}, &i=1 ,\dots, K \\
    &\displaystyle&\sum_{j=1}^N \nu_j z_{ji} \leq m^{-1}(y_i)\sum_{j=1}^N z_{ji}, &i=1 ,\dots, K
\end{array}
\end{equation*}
\end{lemma}
\noindent The above linear program, denoted \textbf{LP($H,\mathbf{y}, \mathbf{c}$)}, has $NK$ variables $z_{ji}$. We can conclude that the time complexity of \textbf{LP}($\dist_\delta,\mathbf{y}, \mathbf{c}$) is $O(N^{2.5}K^{2.5})$, i.e. polynomial in the number of partitions ($K$) to represent the piecewise approximation of $\obj(\cdot;\nu)$ and the size ($N$) of the support for the discretized distribution~$\dist_\delta$ (\cite{vaidya_speeding-up_1989}). Hence, we can identify optimal signaling mechanisms for discretized objectives and discretized distributions over state with efficient computation time.

\noindent To use the above linear program, we require $h$ to be piecewise-constant. To implement this, we approximate our Lipschitz objective $\obj(\cdot, \param)$ for each value of $\param \in \paramspace$ by a piecewise constant function $\obj_\tau$, where $\tau$ is another discretization parameter. Precisely, for a fixed $\tau$, we create uniform intervals of length $\frac{1}{\tau}$ over $[0,1]$ with the discretized function taking a constant value over each interval equal to the average of the minimum and maximum over this interval. 
\begin{definition}
The \textit{$\tau$-discretization} $\obj_\tau$ of any continuous function $\obj:[0,1]\times\param\rightarrow \R$ is a piecewise constant function such that for all $k \in [\tau], \param\in\paramspace,x \in[\frac{k-1}{\tau},\frac{k}{\tau})$:
\begin{align*}
h_\tau(x,\param) \coloneqq h(\frac{2k-1}{2\tau},\param)
\end{align*}
\end{definition}
\noindent Consequently, $\obj_\tau$ is piecewise-constant over each interval $[\frac{k-1}{\tau},\frac{k}{\tau})$ for all $k \in [\tau]$.

\subsection{$\epsilon$-optimal Design}
\label{subsec:epsilon_optimal_design}
To extend the solution obtained from Lemma~\ref{lemma:lp_discrete} for the discrete distribution $\dist_\delta$ to the original continuous distribution~$\dist$, we introduce the notion of an \textit{$\epsilon$-optimal} signaling mechanism where $\epsilon > 0$ bounds the suboptimality gap of the mechanism.

\begin{definition}
\label{defn:eps_approx}
A mechanism $\hat{\pi}$ is $\epsilon$-optimal for a problem instance \eqref{eqn:gen_objective_stateful} defined by distribution $\dist$ over $\paramspace$ and utility function $\obj$ if $|V_{\dist,\obj}(\hat{\pi}) - V_{\dist,\obj}(\pi^\ast)| \leq \epsilon$.
\end{definition}

\noindent Thus, the $\epsilon$-optimal signaling mechanism $\hat{\pi}$ must be \textit{close} to $\pi^\ast$ when evaluated according to the planner's objective {in expectation} with the true prior distribution $\dist$. However, the tuples $\pi^\ast$ and $\hat{\pi}$ themselves need not be ``close'' and cannot be compared in a straightforward manner. 

\noindent We utilize~Lemma~\ref{lemma:lp_discrete} to develop an $\epsilon$-optimal signalling mechanism as follows. We first solve for $\pi_{\dist_\delta,\obj_\tau}^\ast = \langle \I_\delta , \{z^\delta_{\nu_j}\}_{\nu_j \in \paramspace_\delta}\rangle$ using the linear program in Lemma \ref{lemma:lp_discrete} as both $\dist_\delta$ is discrete and $\obj_\tau$ is piecewise-constant. We then adapt $\pi_{\dist_\delta,\obj_\tau}^\ast$ to a continuous signaling mechanism $\hat{\pi}_{\dist_\delta,\obj_\tau}~\coloneqq~\langle \I_\delta , \{\hat{z}_\param\}_{\param \in \paramspace}\rangle$ such that for all $j \in [N], \param \in [\nu_{j-1},\nu_j)$, and $\signal \in \I_\delta$, we have $\hat{z}_\param(\signal) \coloneqq z^\delta_{\nu_j}(\signal)$. We prove that, subject to regularity on the distribution of agents' value of in-person work\footnote{This condition is justified by the fact that highly concentrated distribution $G$ can lead to a high sensitivity of equilibrium mass of remote agents to the induced posterior means, thus making approximation difficult.} and the Lipschitz continuity of $\obj$, this solution $\hat{\pi}_{\dist_\delta,\obj_\tau}$ is $\epsilon$-optimal and hence achieves an objective $\epsilon$-close to that of $\pi_{\dist,\obj}^\ast$ (details on the proof are deferred to Sec. \ref{sec:proofs_for_sec_4}). 
\begin{theorem}\label{thm:stateful}
Let $\obj(\fin; \trueparam)$ be uniformly $\eta_{1}$-Lipschitz for all $\trueparam \in \paramspace$ and uniformly $\eta_{2}$-Lipschitz for all $\fin \in [0,1]$. Then, if $\popdist$ is continuously differentiable with $0 < \frac{d\popdist}{dv} \leq \kappa$, the signaling mechanism $\hat{\pi}_{\dist_\delta,\obj_\tau}$ constructed from $\pi_{\dist_\delta,\obj_\tau}^\ast$ obtained by solving \textbf{LP}$(F_\delta, \mathbf{y},\mathbf{c})$ where $y_i = \frac{2i-1}{2\tau}$ and $c_{jk} = h_\tau(y_k;\nu_j)$ (Lemma \ref{lemma:lp_discrete}) is $\epsilon$-optimal for $\delta > \frac{8\eta_{2} +8C\eta_{1} \kappa}{\epsilon}$, $\tau > \frac{4\eta_{1}}{\epsilon}$. 
\end{theorem}
\noindent Although closed-form solutions are not attainable in general for this class of preference models, the preceding theorem shows that asymptotically optimal approximations can be achieved through sufficient discretization. This implies that, for practical preference models that adhere to the Lipschitz condition by exhibiting moderate sensitivity to variations in infectiousness or mass, a straightforward procedure exists for identifying nearly optimal solutions in practice. As illustrated in~Sec.~\ref{subsec:numerical}, our computational approach to tackle this large class of models allows us to consider richer preferences in comparison to other works (e.g.~\cite{de_vericourt_informing_2021}).  Furthermore, by construction, the derived mechanisms use identical randomization over signals for each $\trueparam$ in each discretized interval (i.e. interval-based mechanisms). In practice, this interval-based feature of the computed mechanism has the advantage of greater interpretability when discretization is minimal.

\section{Computational Study}
\label{subsec:numerical}

In this section, we apply our computational approach to a class of operationally relevant planner preferences, and demonstrate that we asymptotically recover the optimal planner utility. 

\noindent In Section \ref{subsec:dv_comp}, we provide a detailed comparison of our method on the objectives $\obj_{\text{ref}(\lambda)}$ that are investigated in \cite{de_vericourt_informing_2021}:
\begin{align*}
\obj_{\text{ref}(\lambda)}(\fin;\trueparam) = \lambda \mathbb{E}_{v \sim G}[v\mathbb{I}\{v\geq G^{-1}(\fin)\}]- (1-\lambda) \trueparam (1-\fin)^2.
\end{align*} 
We demonstrate that our approach can recover their optimal closed-form solutions faster than the rate described in Theorem \ref{thm:stateful}. Todemonstrate that our approach is more general, we consider another class of preference for which there is no provably known optimal solution. Particularly, we allow $\dist$ to now be continuous. Moreover, while it is possible to characterize the optimal signaling mechanism in closed-form for specific state distributions (e.g., binary valued), in practice the planner may also want to ensure that the induced in-person mass is not too close to fully remote or full in-person work. This additional ``regularization'' becomes especially relevant for hybrid work settings in which the workplace facilities need to be used at moderate occupancy levels to contain the risk of transmission and yet maintain sufficient productivity levels. To reflect this preference, we modify the planner's utility function: 
\begin{align}
\label{eq:h_rho_eqn}
\obj_\rho(\fin;\trueparam) = \frac{1}{2}\big((1-\rho)(5(1-\fin^2) -  \trueparam (1-\fin)^2)\big)  + \rho \fin(1-\fin), 
\end{align}
where $\rho \fin(1-\fin)$ reflects the regularization term with parameter~$\rho$ and other terms are same as for~$\obj_{\text{ref}(0.5)}(\fin;\trueparam)$. Previously known results cannot be used to compute an optimal mechanism for such a preference model due to its complex dependence on posterior means. However, the model satisfies the conditions of Theorem~\ref{thm:stateful}; thus, we can design an $\epsilon$-optimal mechanism using the LP-based solution introduced in~Sec.~\ref{sec:stateful}. Furthermore, we can bound the approximation loss in terms of discretization parameters~($\delta, \tau$).

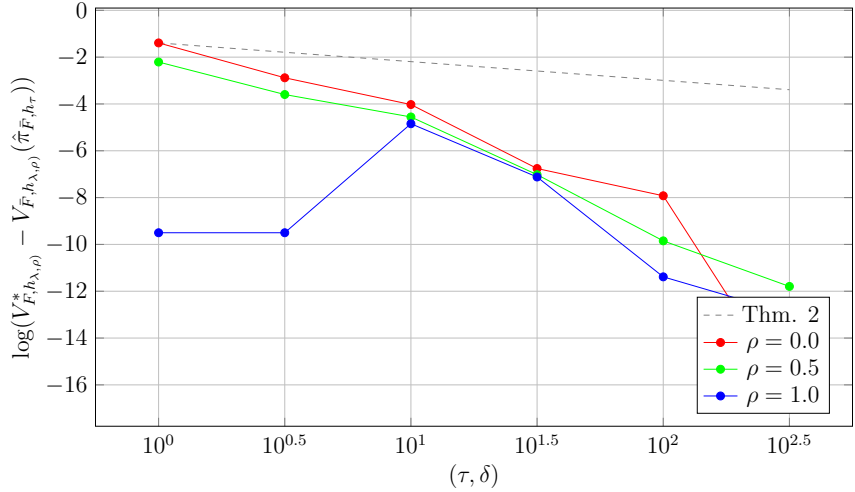
\begin{figure}[h!]

\centering
\begin{tikzpicture}[scale=0.75]
\begin{axis}[
symbolic x coords={1,3,10,32,100,316},
xticklabels={$10^{0}$,$10^{0.5}$,$10^{1}$,$10^{1.5}$,$10^{2}$,$10^{2.5}$},
xtick=data,
height=9cm,
width=15cm,
grid=major,
xlabel={$(\tau,\delta)$},
ylabel={$\log(V_{F,h_{\lambda,\rho)}}^\ast - V_{\bar{F},h_{\lambda,\rho)}}(\hat{\pi}_{\bar{F},h_{\tau}}))$},
legend style={
cells={anchor=east},
legend pos=south east,
}
]
\addplot[-,dashed,gray] coordinates {
(1,-1.39188)  (3,-1.79188)  (10,-2.19188)  (32,-2.59188)  (100,-2.99188)  (316,-3.39188)};
\addplot[red,mark=*] coordinates {
(1,-1.39188)  
(3,-2.88172)  
(10,-4.02515)  
(32,-6.75781)  
(100,-7.92214)
(316,-16.27845)
};
\addplot[green,mark=*] coordinates {
(1,-2.20855)  
(3,-3.59715)  
(10,-4.55526)  
(32,-7.02078)  
(100,-9.84782)
(316,-11.79503) 
};
\addplot[blue,mark=*] coordinates {
(1,-9.50138)  
(3,-9.50127)  
(10,-4.84466)  
(32,-7.11988)  
(100,-11.38512) 
(316,-12.96203)
};

\legend{Thm. \ref{thm:stateful},$\rho = 0.0$,$\rho = 0.5$,$\rho = 1.0$}
\end{axis}
\end{tikzpicture}

\caption{Error of computed $\epsilon$-optimal solution $\hat{\pi}_{F_\delta,h_{\tau}}$ as discretization ($\delta,\tau$; assume equal) increases and the regularizations $\rho$ are varied; $\lambda = 0.5$.} \label{fig:continuous_error}
\end{figure} 
\noindent We leverage our result in Theorem \ref{thm:stateful} and assume that the true optimal signaling mechanism $\pi_{F,h_\rho}^\ast$ is well approximated by choosing the approximate solution $\hat{\pi}_{F_\delta,h_\tau}$ for $\delta=\tau = 1000$. Figure~\ref{fig:continuous_error} shows how the planner's utility corresponding to our approximate solution compares against the computationally obtained optimal value for varying levels of discretization and choice of regularization parameter. Observe that the error of our computational solution reduces quickly to the limits of numerical precision and achieves much faster rate of convergence to~$0$, in comparison to the theoretically guaranteed rate of~$\frac{1}{\epsilon^5}$. Again, this can be explained by noting that our computed solution has an interval-based structure -- in particular, the agent distribution induced by our signaling mechanism is (close to) the outcome achieved by no- or full-information mechanisms, depending on the value of underlying state (see Fig. \ref{fig:y_theta_plot}). Since interval-based mechanisms are extreme points of the polytope containing all signaling mechanisms, our LP-based solution achieves a very fast convergence rate in traversing the extreme points of the polytope (\cite{bergemann_information_2019}). 

\noindent Finally, we can verify that the optimal design for the preference model~$\obj_\rho$ induces outcomes with progressively more moderate in-person agent mass as $\rho$ increases. In Figure \ref{fig:y_theta_plot}, we plot the joint distribution of the equilibrium mass and the state corresponding to the approximately optimal mechanism $\hat{\pi}_{F_\delta,h_\tau}$ with the $\epsilon$-optimality guarantee provided in Theorem \ref{thm:stateful} ($\delta = \tau = 1000$). From these plots we can determine which equilibrium masses are most frequent by considering the marginal distribution. These plots also reveal the structure of the signalling mechanism since we can determine which posterior means (through the equilibrium mass) are mapped from each state. We again consider $\obj_\rho$ from Equation \eqref{eq:h_rho_eqn} for values of $\rho \in [0.5,0.75,1]$ and qualitatively compare how the equilibrium behavior compares. We again choose $\popdist\sim Unif[0,6]$ and $\dist\sim Unif[0,10]$, and consider the product distribution over $(\trueparam,\fin_{\pi_{F,h}^\ast}(\signal))$. 
\begin{figure}[h!]
\centering
\includegraphics[width=125mm]{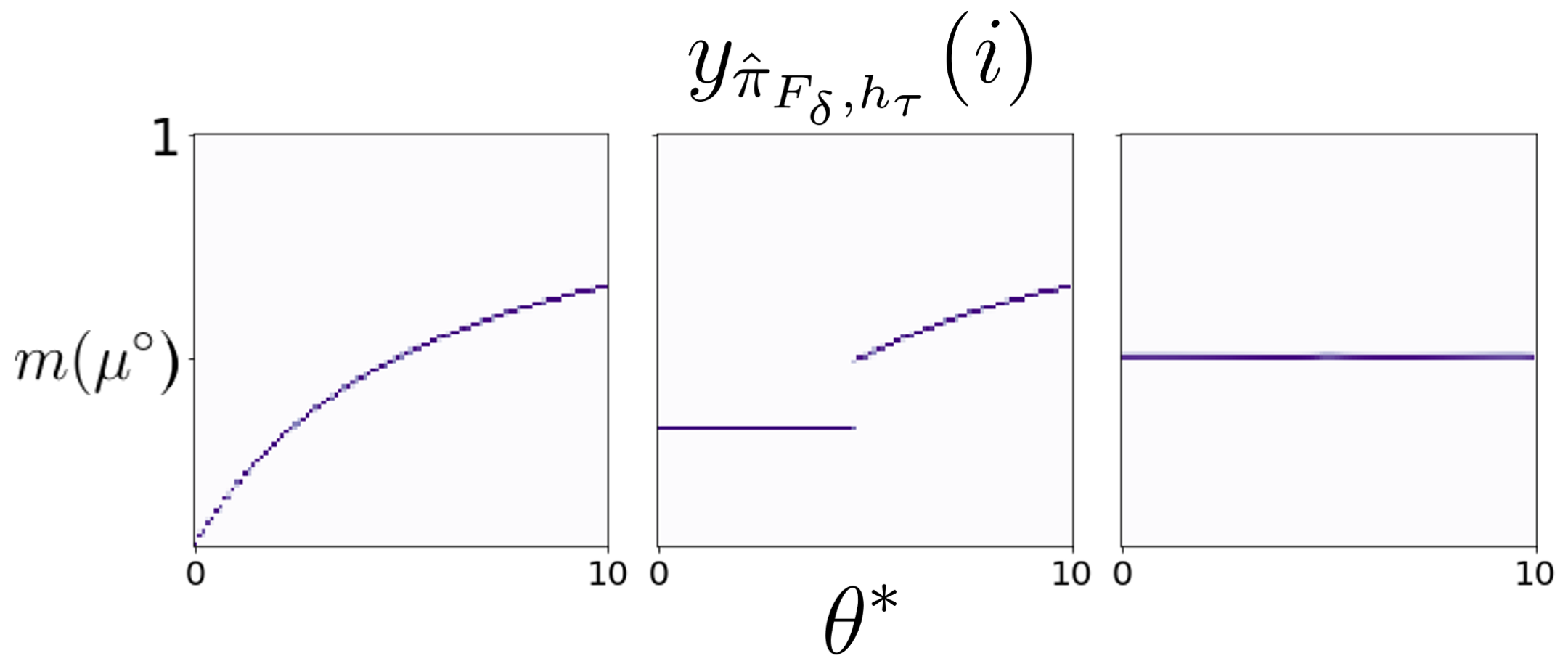}
\caption{Density plot over the joint distribution of $\big(\trueparam,\fin_{\pi_{\dist,\obj}^\ast}(\signal)\big)$ (higher density is in purple).\\ $\rho = 0.5 \text{ (left)}, 0.75\text{ (middle)}, 1\text{ (right)}$.}
\label{fig:y_theta_plot}
\end{figure}

\noindent In Fig. \ref{fig:y_theta_plot}, observe that for $\rho = 0.5$ (and for all $\rho < 0.5$ though not shown), all the probability mass lies on an approximately smooth curve that is identical to the structure of $m(\theta)$. This implies that, at optimality, in any given state $\trueparam$, the planner simply reveals the state and the equilibrium remote mass is $m(\trueparam)$. Hence, fully-informative signaling is optimal for $\rho \leq 0.5$.\\
\noindent We also observe that as $\rho$ increases to $1$ and the regularization term becomes more prominent, the planner strongly prefers moderate masses and seeks to avoid equilibrium masses close to 0 or 1. The prior mean belief induces a moderate mass ($m(\priormean) = \frac{5}{11}$), so the optimal mechanism shifts to the non-informative mechanism as that generates this same moderate posterior mean regardless of state. \\
Interestingly, we find that for intermediate values such as $\rho=0.75$, we obtain a mechanism that directly reveals the state for larger values when the function $m(\cdot)$ levels off, but elects to aggregate the states for smaller values of the state to a single signal. This mechanism can be thought of as a combination of the non-informative and fully-informative mechanism as agents learn whether or not $\trueparam$ exceeds a threshold and then are either directly revealed the state or revealed nothing further. This mechanism belongs to a class of mechanisms that either only reveals the interval containing the state or fully reveals the state on each interval for a partition of the state space. This class falls under a superset of interval-based mechanisms which captures the extreme points of the polytope containing all signaling mechanisms. In fact, such mechanisms find prominence in optimal designs and the related literature for simpler settings than we consider (\cite{ivanov_optimal_2015,guo_interval_2019}). This reinforces that computationally-obtained signalling mechanisms from Theorem \ref{thm:stateful} allow us to achieve near-optimal designs when analytical approaches to design becomes intractable. 
Consequently, for many practical objectives, planners may need to consider mechanisms that do not have MPS to achieve good outcomes. 

\noindent While certain planner preferences recover results similar to that of~\cite{de_vericourt_informing_2021} where optimal mechanisms are based on no- or full-information signaling mechanisms, we also demonstrated examples where partially informative, interval-based signaling mechanisms are optimal (Fig. \ref{fig:y_theta_plot}). Our approach handles general forms of state uncertainty over a continuous domain, accommodates a richer class of planner preferences and identifies signalling mechanisms with a more complex structure. Our approach also provides an efficiently computable solution, with the approximation error converging to zero even with a coarse discretization level.

\section{Concluding Remarks}
\noindent In this paper, we introduced a model to study information provision for strategic hybrid workers. The central planner seeks to control the mass of in-person workers across each group in the equilibrium outcome. Our model captures two key features: (a) a general objective that aims to maximize the probability that the equilibrium outcome lies in a particular set which may or may not be state-dependent; (b) heterogeneous workers making strategic decisions to trade-off in-person work and infectious risk. We provided a complete description of the equilibria of the game in response to the signals and derived the optimal signaling mechanism that the planner can employ.

\noindent For settings with more complex objectives, we derived algorithms that compute $\epsilon$-optimal signaling information disclosure rules. These analytic and numerical insights suggest that simple information disclosure rules using interval-based disclosure strategies, which are more easily implementable in practice, are sufficient to achieve near-optimal or optimal outcomes. While it is not always possible to codify exact functional representations of the objectives or functions in practice, these insights can inform how planners should strategically disclose information.

\noindent These results provide valuable guidelines for the design and deployment of signaling mechanisms, especially as hard intervention measures are being phased down by public health agencies.
\bibliographystyle{informs2014} 
\bibliography{ref}
\appendix
\section{Appendix}
\subsection{Proofs for Sec. \ref{sec:static_id}}
\label{appendix:proofs_sec_iii}
\subsubsection*{Proof of Lemma \ref{lemma:num_signals_bdd}}
For any optimal signaling mechanism $\pi^\ast = \langle \I, \{z_\param\}_{\param\in\paramspace}\rangle$, consider its direct mechanism which takes the form
$\mathcal{T}_\pi^{\ast} = \{(q_\signal,\posterior{\signal})\}_{\signal \in \I}$. Suppose that for any $k \in [K]$ there exists $i,j\in\I$ such that $\posterior{i},\posterior{j} \in \bar{\paramspace}_k$. Then, replacing the two signals $i$ and $j$ with a single signal $ij$, consider $\pi'=\langle \I\cup\{ij\}\setminus\{i,j\}, \{z_\param\}_{\param\in\paramspace}\rangle$ where $z_\param(s) = z_\param(s)$ for all $s \in \I\setminus\{i,j\}, \param \in \paramspace$, and $z_\param(ij) = z_\param(i) + z_\param(j)$ for all $\param \in \paramspace$. Let $\posterior{i}\leq\posterior{j}$ without loss of generality. Observe that $q_s'$ and $\posterior{s}'$ are unchanged for all $s\in\I$ such that $s\neq i,j$. Moreover, from \eqref{eqn:q_i} and \eqref{eqn:theta_i}, observe that $q_{ij}' = q_i+q_j$ and $\posterior{ij}' = \frac{q_i}{q_i+q_j}\posterior{i}+\frac{q_j}{q_i+q_j}\posterior{j}$ is a weighted average of $\posterior{i}$ and $\posterior{j}$ so $\posterior{i}\leq\posterior{ij}'\leq\posterior{j}$ and $\posterior{ij}'\in\bar{\paramspace}_k$. Hence, from \eqref{eqn:red_1_form}, observe that $\mathcal{T}_{\pi'}$ achieves the same objective as $\mathcal{T}_\pi^{\ast}$ and therefore must also be optimal. Likewise, if $\posterior{i},\posterior{j} \notin \bar{\paramspace}_k$ for all $k$, then from examination of \eqref{eqn:red_1_form}, we note that $\mathcal{T}_{\pi'}$ is also optimal since it achieves an objective no smaller than that achieved by $\mathcal{T}_{\pi^\ast}$:
\begin{align*}
 \sum_{s\in\I}\sum_{k=1}^{K} q_s \mathbb{I}\{\lol{k} \leq \posterior{s} \leq \hil{k}\} &= \sum_{s\in\I\setminus\{i,j\}}\sum_{k=1}^{K} q_s \mathbb{I}\{\lol{k} \leq \posterior{s} \leq \hil{k}\}+\sum_{k=1}^{K} q_i \mathbb{I}\{\lol{k} \leq \posterior{i} \leq \hil{k}\}+\sum_{k=1}^{K} q_j \mathbb{I}\{\lol{k} \leq \posterior{j} \leq \hil{k}\} \\
 &= \sum_{s\in\I\setminus\{i,j\}}\sum_{k=1}^{K} q_s \mathbb{I}\{\lol{k} \leq \posterior{s} \leq \hil{k}\} \\
 &= \sum_{s\in\I\setminus\{i,j\}}\sum_{k=1}^{K} q_s' \mathbb{I}\{\lol{k} \leq \posterior{s}' \leq \hil{k}\} \\
&\leq \sum_{s\in\I\cup\{ij\}\setminus\{i,j\}}\sum_{k=1}^{K} q_s' \mathbb{I}\{\lol{k} \leq \posterior{s}' \leq \hil{k}\}
 \end{align*}

\noindent Consequently, recursively performing this reduction in the size of the signal set $\I$ for each interval $\bar{\paramspace}_k$ for $k \in [K]$ and $[0,M]\setminus\cup_{k=1}^K \bar{\paramspace}_k$, we obtain an optimal direct mechanism with each subset containing no more than one posterior mean $\mu_i$. Hence, from this optimal direct mechanism, replacing the signal set with $\I = [K+1]$ where $\posterior{k} \in \bar{\paramspace}_k$ for all $k\in[K]$ achieves the result.\qed 

\subsubsection*{Proof of Lemma \ref{lemma:reduce_MPC_constraints}}
Consider any posterior mean distribution of the form $H(t) =\sum_{k=1}^{K+1}~q_k~\mathbb{I}\{\posterior{k}~\leq~t\}$ for all $t\in\paramspace$. We show the equivalence between $H\succcurlyeq \dist$ and constraints given by \eqref{eqn:MPC} and \eqref{eqn:mean_mean}.\\
\noindent From the definition of mean-preserving contractions, $H\succcurlyeq \dist$ implies the constraints in \eqref{eqn:MPC} are satisfied and that the mean is preserved across $\dist$ and $H$ which subsequently implies  \eqref{eqn:mean_mean}.

\noindent Moreover, consider any distribution $H(t) =\sum_{k=1}^{K+1}~q_k~\mathbb{I}\{\posterior{k}~\leq~t\}$ satisfying \eqref{eqn:MPC} and \eqref{eqn:mean_mean}. Then, by \eqref{eqn:mean_mean}, the mean of $H$ and $\dist$ are equal so $\int_0^1 H^{-1}(s) ds = \int_0^1 F^{-1}(s) ds$. For any $0 \leq x < 1$, there exists $n \in [K]$ such that $\sum_{j=1}^{n-1} q_j\leq x < \sum_{j=1}^{n} q_j$. Observe that $f(t) \coloneqq \int_0^t (F^{-1}(s)-H^{-1}(s)) ds$ is convex over $\sum_{j=1}^{n-1} q_j\leq t \leq \sum_{j=1}^{n} q_j$ since $\int_0^t F^{-1}(s) ds$ is convex and $\int_0^t H^{-1}(s) ds$ is linear over $\sum_{j=1}^{n-1} q_j\leq t \leq \sum_{j=1}^{n} q_j$. Since the constraints of \eqref{eqn:MPC} imply that $f(\sum_{j=1}^{n-1} q_j),f(\sum_{j=1}^{n} q_j)\leq 0$, the convexity of $f$ implies that $f(x) \leq 0$. Hence, $\int_0^x H^{-1}(s) ds \geq \int_0^x F^{-1}(s) ds$ for all $x\in [0,1)$ which implies $H\succcurlyeq \dist$.\qed
\subsubsection*{Proof of Theorem \ref{thm:r123}}
\begin{proof}
\noindent In $\mathrm{R1}$, we know that $\pi_{F,\obj}^\ast=\pi_{\mathrm{NI}}$, and hence $\mathcal T_{NI}=\{(1, \priormean)\}$. From~\eqref{eqn:red_1_form}, we obtain~$V_{F,h}^\ast=1$, which is the maximum achievable value of planner's objective function.

\noindent To proceed with R2 and R3, we can simplify the objective function in~\eqref{eqn:red_1_form} using Lemma~\ref{lemma:num_signals_bdd} as follows:
\begin{align*}
    \sum_{i =1}^{\numintervals+1} \sum_{k=1}^{\numintervals} q_i \mathbb{I}\{\lol{k} \leq \posterior{i} \leq \hil{k}\}=\sum_{i=1}^\numintervals q_i = \left(1-q_{\numintervals+1}\right). 
\end{align*}
Hence, the problem of optimal signaling mechanism design can be expressed as follows:
\begin{align*}
    &\min_{\mathcal{T}_\pi: \{(q_i, \posterior{i})\}_{i \in [\numintervals+1]}}  q_{\numintervals+1}\\
    \text{s.t.} & \quad H\succcurlyeq \dist, \\ 
    & \quad   \posterior{i} \in \bar{\Theta}_i,\quad i\in[\numintervals],  
\end{align*}
where the second constraint ensures implementability by requiring that the posterior distribution $H$ is a mean-preserving contraction of prior distribution~$\dist$. Using~Lemma~\ref{lemma:reduce_MPC_constraints}, the definitions of sets~$\bar{\Theta}_i$, and the fact that all signal probabilities must sum to~$1$, we can re-write the above problem:  
\begin{subequations}\label{eq:equivalentopt}
\begin{align}
 &\min q_{\numintervals+1} \label{eq:qkplus1}\\ 
    \text{s.t.}   &\sum_{i \in [\numintervals+1]} q_i=1 \label{eq:probsum}\\ 
    & \lol{i}\leq \posterior{i} \leq \hil{i}, \quad \forall i\in[\numintervals] \label{eq:posteriorbounds}\\
    & \eqref{eqn:MPC}, \eqref{eqn:mean_mean} \nonumber. 
\end{align}
\end{subequations}
We now proceed to solve for the optimal signaling mechanism for~$\mathrm{R2}$; the proof for~$\mathrm{R3}$ is analogous. For simplicity and without loss of generality, we assume $\maximal=1$ since any optimal solution is invariant to linear scaling.

\underline{Claim 1.} $\posterior{\numintervals+1} \geq \priormean$ and $q_{\numintervals+1} > 0$, hence $q_{\numintervals+1}^\ast > 0$. 

Under~$\mathrm{R2}$, we know that $\hil{\numintervals}<\priormean$. From constraint $\posterior{\numintervals}\leq\hil{\numintervals}$, we obtain $\posterior{\numintervals}\leq \hil{\numintervals}<\priormean$. \\
Suppose that $\posterior{\numintervals+1}<\priormean$. Then by~\eqref{eqn:mean_mean} and by condition of $\mathrm{R2}$, we obtain $\priormean=\sum_{i=1}^{\numintervals+1} q_i \posterior{i}< \sum_{i=1}^{\numintervals+1} q_i \priormean = \priormean$. However, this is a contradiction. Hence, we conclude that  $\mu_{\numintervals+1}\geq\priormean$.\\
Next, suppose that $q_{\numintervals+1}=0$. Then, $\priormean=\sum_{i=1}^\numintervals q_i \posterior{i}\leq \sum_{i=1}^\numintervals q_i \hil{i} \leq \hil{\numintervals} \sum_{i=1}^\numintervals q_i = \hil{\numintervals}$. However, this is a contradiction since under~$\mathrm{R2}$, $\hil{\numintervals}<\priormean$. Hence, we conclude that  $q_{\numintervals+1}>0$ (and hence $q_{\numintervals+1}^\ast > 0$).

\underline{Claim 2.} $q_j^* = 0$ for all $j < \numintervals$.\\ 
Suppose for the sake of contradiction that $q_j^\ast > 0$ for some $j < \numintervals$, and let the corresponding distribution of posterior means be denoted~$H^\ast$. Now consider the new distribution of signal probabilities obtained by decreasing $q_j^\ast$ by $\tfrac{(\posterior{\numintervals+1}^\ast-\posterior{\numintervals}^\ast)}{(\posterior{\numintervals+1}^\ast-\posterior{j}^\ast)}\epsilon$, increasing $q_{\numintervals}^*$ by $\epsilon$, decreasing $q_{\numintervals+1}^*$ by $\tfrac{(\posterior{\numintervals}^\ast-\posterior{j}^\ast)}{(\posterior{\numintervals+1}^\ast-\posterior{j}^\ast)}\epsilon$, for a small $\epsilon>0$. Then, the value of objective~\eqref{eq:qkplus1} strictly increases and constraints~\eqref{eq:probsum} and~\eqref{eq:posteriorbounds} are still satisfied. Furthermore, this is a convex stochastic modification of the original mechanism; hence the modified set of tuples generates a distribution of posterior means such that~$H'\succcurlyeq H^\ast$ (refer to~Theorem 3.A.7. of \cite{shaked_univariate_2007}). By transitivity, we obtain that $H'\succcurlyeq \dist$, which implies that constraints~\eqref{eqn:MPC}, \eqref{eqn:mean_mean} for ensuring mean-preserving contraction also hold. This establishes the contradiction. Hence, $q_j^* = 0$ for all $j < \numintervals$.  

Following the above claims, we can simply rename $q_{\numintervals}$ to $q_1$ and $q_{\numintervals+1}$ to $q_2$, and similarly rename $\posterior{\numintervals}$ to $\posterior{1}$ and $\posterior{\numintervals+1}$ to $\posterior{2}$, where the signal set is $\I = \{1,2\}$. The problem~\eqref{eq:equivalentopt} simplifies as follows: 
\begin{subequations}\label{eq:simpleopt}
\begin{align}
   & \max_{q_1,q_2,\posterior{1},\posterior{2}} q_1 \\
  \text{s.t.}\quad & q_1+q_2=1 \\
   & \lol{\numintervals} \leq  \posterior{1} \leq \hil{\numintervals} \\
   & \priormean \leq \posterior{2} \leq \maximal \\
   & q_1\posterior{1} \geq \int_0^{q_1} \dist^{-1}(s)ds \label{eqn:convv_new}\\ 
   & q_1 \posterior{1} + q_2 \posterior{2} = \priormean. 
\end{align}
\end{subequations}
Suppose that~$(q_1^\ast,q_2^\ast,\posterior{1}^\ast,\posterior{2}^\ast)$ is an optimal solution of ~\eqref{eq:simpleopt} with~$\posterior{1}^* < \hil{\numintervals}$. Then we can find another optimal solution by choosing same signal probabilities~$q_1'=q_1^\ast$, $q_2'=q_2^\ast$, but the posterior means as $\posterior{1}'=\hil{\numintervals}$ and $\posterior{2}'=\priormean+\tfrac{(\posterior-\hil{\numintervals})(\posterior{2}^\ast-\priormean)}{(\priormean-\posterior{1}^\ast)}$ (this follows by noting that all the constraints in~\eqref{eq:simpleopt} are satisfied). Hence, we can restrict~$\posterior{1}=\hil{\numintervals}$ in the optimal design. Following~\eqref{eq:barffn}, we can rewrite the constraint~\eqref{eqn:convv_new} as $0\leq q_1 \leq \bar f(\hil{\numintervals})$ and substitute $q_2$ with $1-q_1$ to obtain:
\begin{subequations}\label{eq:finalopt}
\begin{align}
    &\max_{q_1,\posterior{2}}\quad q_1 \\
    \text{ s.t.   } &q_1 \hil{\numintervals} + (1-q_1) \posterior{2} = \priormean \label{eq:finalopt1}\\
    & 0 \leq q_1 \leq \bar f(\hil{\numintervals}) \label{eq:finalopt2}\\
  & \priormean \leq \posterior{2} \leq M .\label{eq:finalopt3}
\end{align}
\end{subequations}

From~\eqref{eq:finalopt1} we obtain  $q_1=\frac{\posterior{2}-\priormean}{\posterior{2}-\hil{\numintervals}}$ and using $\hil{\numintervals}\leq \priormean \leq \posterior{2} \leq M$ (Claim 1 and~\eqref{eq:finalopt3}), we know that $q_1\leq \frac{\maximal-\priormean}{\maximal-\hil{\numintervals}}$. Combining with~\eqref{eq:finalopt2}, the optimal value of~\eqref{eq:finalopt} is $q_1^\ast=\min\{\bar f(\hil{\numintervals}),\frac{M-\priormean}{M-\hil{\numintervals}}\}$. To summarize, $\posterior{2}^\ast=\frac{\priormean-q_1^\ast \hil{\numintervals}}{1-q_1^\ast}$, $\posterior{1}^\ast=\hil{\numintervals}$, and $q_2^\ast=1-q_1^\ast$, specifies the optimal direct mechanism~$\mathcal T^\ast=\{(q_1^*, \posterior{1}^*),(q_2^*, \posterior{2}^*)\}$.

Finally,  given the optimal objective value $V_{F,\obj}^\ast$ and the direct mechanism $\mathcal{T}^\ast$, we want to find a mechanism $\pi_{F,\obj}^\ast = \langle \I, \{z_\theta\}_{\theta \in \Theta} \rangle$ that implements $\mathcal{T}^\ast$ to achieve the value~$V_{F,\obj}^\ast$. Here, we appeal to earlier results: Prop.~1 in \cite{gentzkow_rothschild-stiglitz_2016} and Thm.~3.A.4 in \cite{shaked_univariate_2007}. In particular, consider the discrete distribution $\mathcal G$ that places probability $q_1^\ast$ on $\posterior{1}^\ast=\hil{\numintervals}$ and $1-q_1^\ast$ on $\posterior{2}^*$. Then, $g(x) \coloneqq \int_{0}^x \mathcal{G}(t) dt$ can be expressed as:  
\begin{align*}
    g(x) &= \begin{cases} 0 & x \leq \hil{\numintervals} \\ q_1^\ast (x-\hil{\numintervals}) & \hil{\numintervals} < x \leq \posterior{2}^\ast \\ q_1^*(\posterior{2}^\ast-\hil{\numintervals})+(x-\posterior{2}^\ast) & \posterior{2}^\ast<x\leq \maximal   \end{cases}
\end{align*}
We know that $g$ is convex and $g(x)\leq g^\circ(x)\coloneqq\int_0^x F(t)dt$ for all $x\in[0,1]$. Moreover, $g'(0) = {g^\circ}'(0)$ and $g'(1) = {g^\circ}'(1)$. From \cite{shaked_univariate_2007} and \cite{ivanov_optimal_2021}, it is known that if there exists $s \in [\hil{\numintervals},\posterior{2}^\ast]$ such that $g^\circ(s) = g(s)$ ---see \underline{Claim 3} below--- then $g$ is tangent to $g^\circ$ at $s$ and therefore is tangent on each linear segment of $g$. It follows that such a direct mechanism can be implemented using a signaling mechanism that has a monotone partitional structure with $t_0 = 0,\; t_1 = s\coloneqq F^{-1}(q_1^\ast),\; t_2 = 1$. To check consistency, note that $\mathbb{P}[\theta \in [0,s]] = F(s) = F(F^{-1}(q_1^\ast)) = q_1^\ast$, which is the probability of signal~$1$. 

\underline{Claim 3.}  $\exists \; s \in [\hil{\numintervals},\posterior{2}^\ast]$ such that $g^\circ(s) = g(s)$.

Suppose by contradiction that for all $s \in [\hil{\numintervals},\posterior{2}^*]$, $g^\circ(s) > g(s)$. Then since $g^\circ-g$ is convex over $[\hil{\numintervals},\posterior{2}^*]$, let $\inf_{t \in [\hil{\numintervals},\posterior{2}^*]} g^\circ(t)-g(t) = \epsilon > 0$ with some minimizer $t^* \in [\hil{\numintervals},\posterior{2}^*]$ such that $g^\circ(t^*)-g(t^*) = \epsilon$. Furthermore, let $\tilde{\posterior{2}^*}$ solve $(q_1^*+\epsilon)(x-\hil{\numintervals}) = q_1^*(\posterior{2}^*-\hil{\numintervals})+(x-\posterior{2}^*)$ and define the function $\tilde{g}$:
\begin{align*}
    \tilde{g}(x) &= \begin{cases} 0 & x \leq \hil{\numintervals} \\ (q_1^*+\epsilon) (x-\hil{\numintervals}) & \hil{\numintervals} < x \leq \tilde{\posterior{2}^*} \\ (q_1^*+\epsilon)(\tilde{\posterior{2}^*}-\hil{\numintervals})+(x-\tilde{\posterior{2}^*}) & x > \tilde{\posterior{2}^*} \end{cases}
\end{align*}
Notice that $\tilde{g}$ is also convex and that $g\leq\tilde{g}\leq f$; hence following \cite{gentzkow_rothschild-stiglitz_2016} the distribution over posterior means with signal probability $q_1^\ast+\epsilon$ on posterior mean $\hil{\numintervals}$ and $1-q_1^*-\epsilon$ on $\tilde{\posterior{2}^*}$ is implementable through a signaling mechanism. However, this would violate the optimality of $q_1^*$, which is a contradiction. 

\end{proof}

\section{Supplementary Information}
\subsection{Proofs for Sec. \ref{sec:model}}
\label{si:proofs_sec_ii}
\subsubsection*{Proof of Proposition \ref{prop:eqbm_formula_prop_1}}
First, we show that \textit{in equilibrium}, there is a critical type $v^*(i) \in \R_+$ such that all agents of type $v \leq v^*(i)$ work remotely, and all agents of type $v > v^*(i)$ work in-person. We denote the equilibrium action of agents with type $v$ by $s_v^*(i)$.
\begin{lemma}
\label{lemma:rectangular_curve}
There exists a critical type $v^*(i) \in \R_+$ such that for all $v \leq v^*(i)$, $s_v^*(i) = \wrr$ and for all $v > v^*(i)$, $s_v^*(i) = \ws$.
\end{lemma}
\begin{proof}
The proof follows by construction. Given an equilibrium in response to observed posterior mean $\posterior{i}$ from generated signal $\signal$, $v^*(i) = \sup\{t:~s_t^*(i)~=~\wrr\}$. Suppose, by contradiction, that there exists $v < v^*(i)$ such that $s_v^*(i) = \ws$. Then, $v \geq c_1(\final(i))\posterior{\signal}+c_2(\final(i))$. However, this would imply that $\hat{v} > c_1(\final(i))\posterior{\signal}+c_2(\final(i))$ for all $\hat{v} > v$ and hence, $s^*_{\hat{v}}(i) = \ws$ and $\sup\{t: s_t^*(i) = \wrr\}\leq v$. This is a contradiction. Thus, we conclude that for all $v \leq v^*(i)$, $s_v^*(i) = \wrr$ and $v^*(i)$ satisfies the conditions of the critical type.\qed
\end{proof}
\noindent Next, we characterize the in-person equilibrium mass in response to signal $i$, and hence the equilibrium remote mass $\final(i)$. 
\begin{lemma}
\label{lemma:compute_inperson_mass}
The equilibrium remote mass $\final(i) = \inf\{u \geq 0: \popdist^{-1}(u) \geq c_1(u)\posterior{\signal}+c_2(u)\}$.
\end{lemma}
\begin{proof}
Let $\final(i) = z$ and $\inf\{u \geq 0: \popdist^{-1}(u) \geq c_1(u)\posterior{\signal}+c_2(u)\} = z+\epsilon$ for some $\epsilon$. By Lemma \ref{lemma:rectangular_curve}, all agents of type $v \leq \popdist^{-1}(z)$ are such that $v \leq c_1(z)\posterior{\signal}+c_2(z)$ and all agents of type $v > \popdist^{-1}(z)$ are such that $ v > c_1(z)\posterior{\signal}+c_2(z)$.

\noindent Suppose by contradiction $\epsilon > 0$. Since $c_1(u)\posterior{\signal}+c_2(u)$ is strictly decreasing in $u$, by definition of infimum, for all $t < z+\epsilon$, $\popdist^{-1}(t) < c_1(t)\posterior{\signal}+c_2(t)$. Hence, $\popdist^{-1}(z+\frac{\epsilon}{2}) < c_1(z+\frac{\epsilon}{2})\posterior{\signal}+c_2(z+\frac{\epsilon}{2})$. But $\popdist^{-1}(z+\frac{\epsilon}{2}) \geq \popdist^{-1}(z)$ since $\popdist^{-1}$ is non-decreasing, and hence $\popdist^{-1}(z+\frac{\epsilon}{2}) > c_1(z)\posterior{\signal}+c_2(z) > c_1(z+\frac{\epsilon}{2})\posterior{\signal}+c_2(z+\frac{\epsilon}{2})$. This is a contradiction.

\noindent Analogously, suppose by contradiction, that $\epsilon < 0$. Then, $\popdist^{-1}(z+\frac{\epsilon}{2}) \leq \popdist^{-1}(z) \leq c_1(z)\posterior{\signal}+c_2(z) < c_1(z+\frac{\epsilon}{2})\posterior{\signal}+c_2(z+\frac{\epsilon}{2})$. 
But, by infimum definition, $\popdist^{-1}(z+\frac{\epsilon}{2}) \geq c_1(z+\frac{\epsilon}{2})\posterior{\signal}+c_2(z+\frac{\epsilon}{2})$. 
This is again a contradiction.

\noindent This implies that $\epsilon = 0$, so $\final(i) = \inf\{u \geq 0: \popdist^{-1}(u) \geq c_1(u)\posterior{\signal}+c_2(u)\}$.\qed 
\end{proof}
Together, Lemma \ref{lemma:rectangular_curve} and Lemma \ref{lemma:compute_inperson_mass} imply Proposition \ref{prop:eqbm_formula_prop_1}.\qed 

\subsubsection*{Proof of Lemma
\ref{lemma:m_smooth}}
\noindent Observe that $0 \leq m(\posterior{}) \leq 1$, since by definition (Prop. \ref{prop:eqbm_formula_prop_1}), $m(\posterior{}) \geq 0$ and $G^{-1}(1) > 0 = c_1(1) \posterior{} + c_2(1)$ which implies that $m(\posterior{}) \leq 1$. Hence, $m(\posterior{})$ is bounded.
\noindent Similarly, letting $f(u) = \frac{\popdist^{-1}(u) - c_2(u)}{c_1(u)}$, we equivalently have $m(\posterior{}) = \inf\{u: f(u) \geq \posterior{} \}$. Since $c_1(u)$ is a strictly decreasing function in $u$ and $\popdist^{-1}(u) - c_2(u)$ is a non-decreasing function, $f(\cdot)$ is strictly increasing. For any $\posterior{}' \leq \posterior{}''$ notice $\{u: f(u) \geq \posterior{}'' \} \subseteq \{u: f(u) \geq \posterior{}' \}$, so $m(\posterior{}') \leq m(\posterior{}'')$ and hence $m$ is non-decreasing. Applying Berge's Maximum Principle, we obtain that $m(\posterior{})$ is continuous. \qed

\subsection{Regime with non-MPS}
\label{subsec:state_indpt_non_mps}
We now focus on the characterization of optimal signaling mechanism in regime~$\mathrm{R4}$, which corresponds to the case when the prior mean~$\priormean$ does not lies in any of the intervals $\bar{\Theta}_k,\; k=1,\dots,\numintervals$, but lies in the gap between two contiguous intervals, i.e., $\exists k'\in[\numintervals]$ such that $\priormean\in( \hil{k'}, \lol{k'+1})$. In this regime, the planner seeks to design signaling that induces posterior mean beliefs outside of the interval~$( \hil{k'}, \lol{k'+1})$, where~$\bar{\Theta}_{k'}=[\lol{k'},\hil{k'}]$ (resp. $\bar{\Theta}_{k'+1}=[\lol{k'+1},\hil{k'+1}]$) is the interval immediately left (resp. right) to the~$\priormean$. 

The following example illustrates that regime~$\mathrm{R4}$ may not necessarily admit an optimal signaling mechanism with MPS:
 \begin{example}
 \label{ex:noMPSinR4}
Let $\dist$ be uniform on $[0,1]$ ($\priormean = 0.5$). Consider $K=2$ and for some small $\epsilon$, define $\bar{\Omega}_1 = [0.4-\epsilon,0.4+\epsilon]$ and $\bar{\Omega}_2 = [0.6-\epsilon,0.6+\epsilon]$. Hence, regime~$\mathrm{R4}$ is active. We can exhaust mechanisms with MPS by considering three cases: Firstly, no information mechanism (which has a MPS with $t_1 = 1$) yields~$0$ planner objective value~\eqref{eqn:red_1_form}.  Secondly, when~$t_1 < 0.8-2\epsilon$, the posterior mean~$\posterior{1}$ is strictly less than $0.4 -\epsilon$ (and thus $\posterior{1} \notin \bar{\Omega}_1\cup\bar{\Omega}_2$), yielding the objective value of $1-F(t_1) < 1$. Finally, when $0.8-2\epsilon \leq t < 1$, then $\posterior{k} \notin \bar{\Omega}_1\cup\bar{\Omega}_2$ for all $k > 1$, so the objective value is $F(t_1) < 1$. Hence, no signaling mechanism with MPS can achieve the maximum objective value of~$1$. \\
However, observe that by choosing $\I = \{1,2\}$ with $z_{\param}(1) = 0.7$ and $z_{\param}(2) = 0.3$ for all $\param \leq 0.5$, and $z_{\param}(1) = 0.3$ and $z_{\param}(2) = 0.7$ for all $\param \geq 0.5$, the posterior means are $\posterior{1} = 0.4$ and $\posterior{2} = 0.6$.  This interval-based mechanism does not have a MPS but yields the maximum objective value of $1$.$\hfill{\vartriangleleft}$
\end{example}  

Hence, unlike regimes~$\mathrm{R1}-\mathrm{R3}$, we need to broaden the search for optimal signaling mechanism to include mechanisms with a non-MPS structure. Analogous to the proof of Theorem~\ref{thm:r123}, we can state the problem of computing an optimal direct signaling mechanism~$\mathcal{T}_\pi^\ast$ as follows:
\begin{align*}
V_{F,h}^\ast &=  \min_{\substack{\mathcal{T}_\pi: \{(q_i, \posterior{i})\}_{i \in [\numintervals+1]} \\ H \succcurlyeq \dist \\\posterior{i} \in \bar{\Theta}_i,\; \forall i\in[\numintervals]} } q_{\numintervals+1},
\end{align*}
where, by definition of $\bar{\Theta}_i$s, we know that $\posterior{k}$ is increasing for all $k \in [\numintervals]$. However, in contrast to both $\mathrm{R2}$ and $\mathrm{R3}$ (where we know that $\posterior{\numintervals} < \posterior{\numintervals+1}$), we can no longer determine how $\mu_{\numintervals+1}$ positioned relative to the other posterior means $\posterior{1}, \dots, \posterior{\numintervals}$. Still, one can computationally solve for~$\mathcal{T}_\pi^\ast$ by iterating over all $\numintervals+1$ possible placements of $\posterior{\numintervals+1}$ relative to $\{\posterior{i}\}_{i\in[\numintervals]}$. That is, we can solve~$\numintervals+1$ individual optimization problems of the form~\eqref{eq:indopt} and finally obtain the optimal direct mechanism as $\mathcal T^\ast= \argmin_{j=1,\dots, \numintervals+1}V_{F,h}^j$.
\begin{subequations}\label{eq:indopt}
\begin{align}
& V_{F,h}^j =  \min_{\substack{\mathcal{T}_\pi: \{(q_i, \posterior{i})\}_{i \in [\numintervals+1]}}} q_{\numintervals+1} \label{eq:indopt1}\\
\text{s.t.} & \quad  H \succcurlyeq \dist \label{eq:indopt2}\\
&\quad  \posterior{i} \in \bar{\Theta}_i,\; \forall i\in[\numintervals] \label{eq:indopt3}\\
& \quad  \begin{cases}
\posterior{\numintervals+1} < \posterior{1} & \text{if} \quad j=1 \\ 
\posterior{j-1} < \posterior{\numintervals+1} \leq \posterior{j}& \text{if} \quad  j=2,\dots,\numintervals \\ 
\posterior{\numintervals+1} > \posterior{\numintervals} & \text{if} \quad j=\numintervals+1  
\end{cases} \label{eq:indopt4}
\end{align}
\end{subequations}
Analogous to the proof of Theorem~\ref{thm:r123}, we can rewrite~\eqref{eq:indopt1}-\eqref{eq:indopt3} and obtain for $j~=~1,\dots,\numintervals~+~1$:  
\begin{align*}
 V_{\dist,\obj}^j =&\min q_{\numintervals+1}         \\
 \text{s.t. }&\sum_{i\in [\numintervals+1]} q_i = 1 \nonumber\\
 & \lol{i} \leq \posterior{i} \leq \hil{i} \quad \forall i \in [\numintervals] \\
 & \eqref{eqn:MPC}, \eqref{eqn:mean_mean},\; \text{and} \eqref{eq:indopt4} \; 
\end{align*}
In fact, each of these problems can convexified by introducing variables $z_j \coloneqq q_j\posterior{j}$. Hence, we arrive at the following result: 
\begin{proposition} An optimal direct signaling mechanism~ $\mathcal{T}_{\pi^*}$ and the corresponding optimal value~$V_{F,h}^\ast$ can be obtained by solving~$\numintervals$ convex programs.  
\end{proposition}

Although one can computationally obtain an optimal direct mechanism~$\mathcal{T}_{\pi^*}$, the question of analytical characterization of optimal signaling mechanism~$\pi_{F,h}^\ast$ that implements~$\mathcal{T}_{\pi^*}$ is not trivial for regime~$\mathrm{R4}$, mainly because we can no longer utilize the MPS. However, as shown in Example~\ref{ex:noMPSinR4}, one can still hope to find an optimal mechanism~$\pi_{F,h}^\ast$ that partitions the statespace~$\paramspace$ into subintervals and uses a fixed probability distribution $z_{\param}(\cdot)$ within each interval. We now show that optimal $\pi_{F,h}^\ast$ has an interval-based structure when we restrict attention to a (more likely) subcase of regime~$\mathrm{R4}$. 

From the definition of~$\mathrm{R4}$, one can intuitively argue that signaling is less effective for a prior distribution~$\dist$ if the induced beliefs are ``tightly concentrated'' in the interval~$(\hil{k'}, \lol{k'+1})$. In particular, for a given~$\dist$, consider the increasing functions: $\ubar{s}(t) = \E_\dist[\param|\param \leq t]$ and $\bar{s}(t) = \E_\dist[\param|\param \geq t]$. Then the information design may not increase planner's expected utility beyond the no-information benchmark if, for most values of~$t$, $\bar{s}(t)$ is close to $\priormean$ and $F(t)$ close to $0$, or $\ubar{s}(t)$ is close to $\priormean$ and $F(t)$ close to $1$. We now introduce a sub-regime of~$\mathrm{R4}$, denoted $\mathrm{R4a}$, which corresponds to the situations when~$\dist$ places sufficient probability mass outside the interval~$( \hil{k'}, \lol{k'+1})$:  

\begin{itemize}
    \item[]  $\mathrm{(R4)}$: $\priormean \notin \cup_{k=1}^K \bar{\Theta}_k$ and $\inf \cup_{k=1}^K\bar{\Theta}_k < \priormean < \sup \cup_{k=1}^K\bar{\Theta}_k $ 
    \item[] $\mathrm{(R4a)}$: $\exists \theta', \theta'' \in \cup_{k=1}^K \bar{\Theta}_k$ such that for any $\delta \in [0,1]$ and any $t \in \paramspace$, the following constraints hold: 
    $0 \leq \delta p(t,\theta') \leq 1$ and $\delta p(t,\theta')+(1-\delta) p(t,\theta'') = 1$, where $p(t,\theta) \coloneqq \frac{(1-F(t))(\bar{s}(t)-\theta)}{F(t)(\theta-\ubar{s}(t))}, \; \theta\in\{\theta', \theta''\}$. 
\end{itemize}

\begin{proposition}
\label{prop:stateless_r4a}
Let~$\theta', \theta'\in \cup_{k=1}^K \bar{\Theta}_k'$ satisfy the conditions for $\mathrm{(R4)}$-($\mathrm{R4a}$), and let $\lambda = \delta p(t,\theta')$ for some $\delta \in [0,1]$, $t \in \paramspace$. Then $V_{F,\obj}^* = 1$ and $\pi_{F,\obj}^* = \langle \{1,2\}, \{z_\param\}_{\param \in \paramspace} \rangle$ with $z_\param(1) = \lambda \mathbb{I}\{\param \in [0, t]\}+\delta \mathbb{I}\{\param \in (t, M]\}$ and $g_\param(1) = (1-\lambda) \mathbb{I}\{\param \in [0, t]\}+(1-\delta) \mathbb{I}\{\param \in (t, M]\}$.
\end{proposition}
\begin{proof}
The proof follows from construction. We can check that $\pi_{F,\obj}^*$ induces $\mathcal{T}_{\pi_{F,\obj}^*} = \{(\lambda F(t) +  \delta(1-F(t)),\theta'),((1-\lambda) F(t) +  (1-\delta) (1-F(t)), \theta'')\}$ and from $\eqref{eqn:red_1_form}$ we conclude that $V_{F,\obj}^\ast = 1$. \qed
\end{proof}

The interval-based structure of the optimal mechanism~$\pi_{F,\obj}^*$ in~$\mathrm{R4a}$ is illustrated in~Fig.~\ref{fig:result_2}: the mechanism is based on a threshold~$t$ which splits the statespace~$\paramspace$ into two disjoint intervals, each corresponding to a signal distribution. Thus, the set of signal~$\I=\{1,2\}$. If $\theta\leq t$, $\pi_{F,\obj}^*$ reveals signal~$1$ with probability~$\lambda$ and signal~$2$ with probability~$(1-\lambda)$. If  $\theta> t$, $\pi_{F,\obj}^*$ reveals signal~$1$ with probability~$\delta$ and signal~$2$ with probability~$(1-\delta)$. Thus, for signal~$1$, the signal probability and induced posterior mean are~$\lambda F(t) +  \delta(1-F(t))$ and $\theta'$, respectively. Similarly, for signal~$2$, these quantities are $(1-\lambda) F(t) +  (1-\delta) (1-F(t))$ and $\theta''$.

\begin{figure}[h]     
    \centering
\includegraphics[width=50mm]{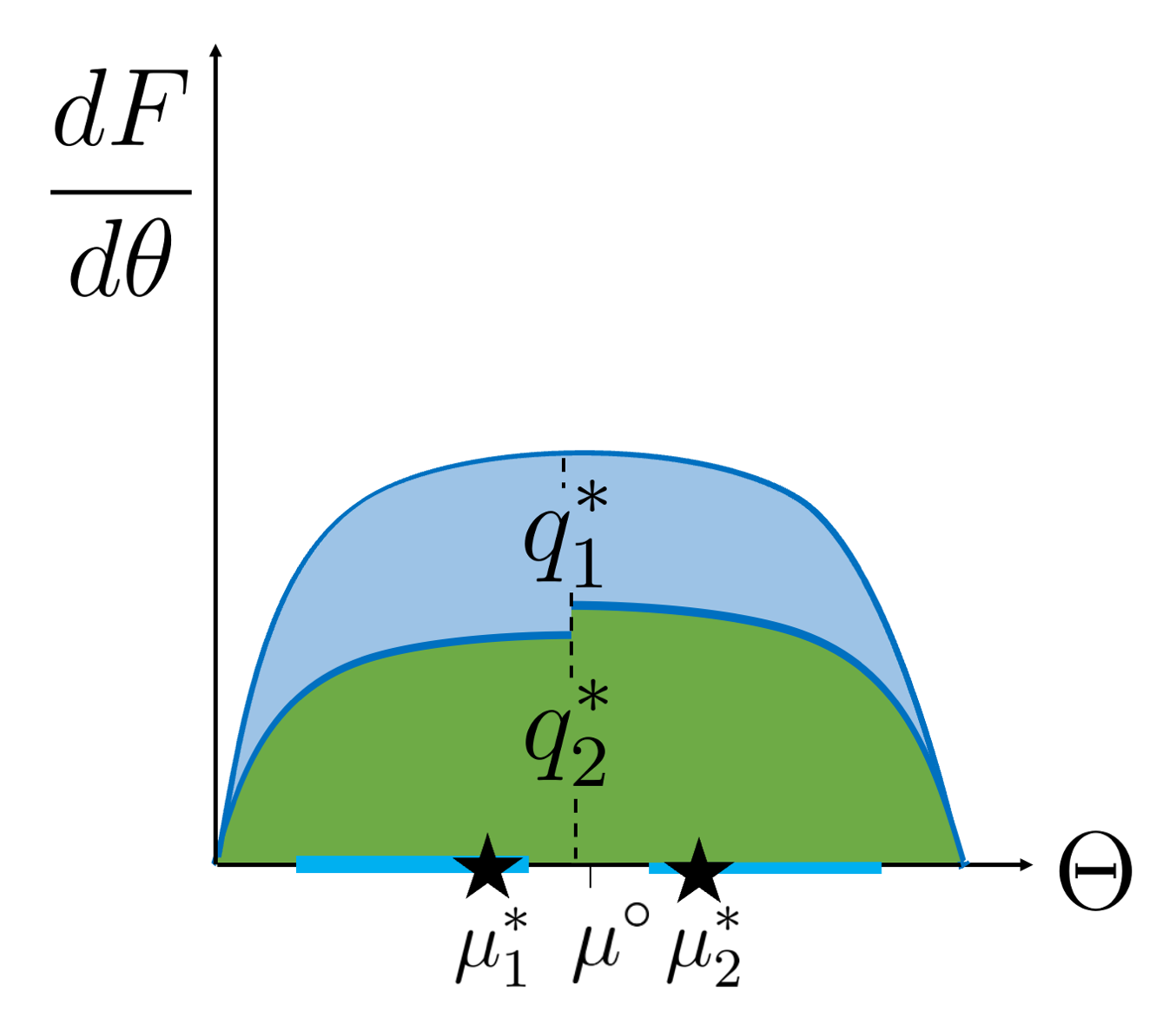}
\caption{The pdf for $\trueparam$ and the intervals $[\lol{k},\hil{k}]$ (blue lines) satisfying conditions for $\mathrm{(R4)}$-($\mathrm{R4a}$). The set of probability distributions for signals~$1$ and $2$ are denoted in blue and green respectively; the induced posterior means~$\posterior{1}$ and $\posterior{2}$ are marked as~$\star$.}\label{fig:result_2}
\end{figure}


\subsection{Proof of Lemma \ref{lemma:lp_discrete}}
\noindent For the discrete distribution, the optimization (see \eqref{eqn:gen_objective_stateful}) can be reformulated as: 
\begin{align*}
V_{H,\obj}^\ast &= \max_{\langle \{z_\param\}_{\param \in \{\nu_j\}_{j=1}^{N}},\I\rangle} \mathbb{E}[h(m(\posterior{i}),\trueparam)]\\
&= \max_{\langle \{z_\param\}_{\param \in \{\nu_j\}_{j=1}^{N}},\I\rangle} \sum_{j=1}^N \sum_{s\in|\I|} \mathbb{P}\{\signal=s,\trueparam=\nu_j\}\obj(m(\posterior{s}),\nu_j)\\
&= \max_{\langle \{z_\param\}_{\param \in \{\nu_j\}_{j=1}^{N}},\I\rangle} \sum_{j=1}^N \sum_{s\in|\I|} p_j z_{\nu_j}(s)\obj(m(\posterior{s}),\nu_j)
\end{align*}
Analogous to Lemma~\ref{lemma:num_signals_bdd}, the objective value remains unchanged when any two signals $s_1,s_2$ with $\posterior{s_1}, \posterior{s_2} \in [m^{-1}(y_{k-1}),m^{-1}(y_k))$ for any $k \in [K]$ are consolidated to a signal $s$ ($\posterior{s_1} \leq \posterior{s} \leq \posterior{s_2}$). Therefore, we define $\I = [K]$ with $m^{-1}(y_{\signal-1})\leq\posterior{\signal}<m^{-1}(y_\signal)$ for all $\signal \in \I$.
Choosing the parameterization $z_{ji} = p_j z_{\nu_j}(i)$, we can rewrite the above optimization problem as:
\begin{equation*}
\begin{array}{ll@{}ll}
\text{maximize}  & \displaystyle&\sum_{j=1}^N \sum_{i=1}^{K} c_{ji} z_{ji} &\\
\text{subject to}& \displaystyle&\sum_{i=1}^{K}  z_{ji} = p_j,  &j=1 ,\dots, N\\
   &\displaystyle &z_{ji} \geq 0,\hspace{0.5cm} &j=1 ,\dots, N,\quad i=1 ,\dots, K \\
    &\displaystyle m^{-1}(y_{i-1})& \leq  \posterior{i} = \frac{\sum_{j=1}^N \nu_j z_{ji}}{\sum_{j=1}^N z_{ji}} \leq m^{-1}(y_i), &i=1 ,\dots, K
\end{array}
\end{equation*}
Expanding the last inequalities we arrive at the specified linear program with optimal solution $\pi_{H,\obj}^\ast = \langle \I, \{z_{\param}\}_{\param \in \discpspace} \rangle$ where $\I = [K]$ and, for all $\signal \in \I$ and $j \in [N]$ setting $z_{\nu_j}(\signal) = 0$ if $p_j = 0$ and $z_{\nu_j}(\signal) = \frac{z_{ji}^\ast}{p_j}$.\qed

\subsection{Proof of Theorem \ref{thm:stateful}}
\label{sec:proofs_for_sec_4}
\noindent To analyze $\hat{\pi}_{\dist_\delta,h_\tau}$, we first note that by the Lipschitz continuity in $\fin$,  $\|h(;\theta)-h_\tau(;\theta)\|_{\infty} \leq \frac{\epsilon}{4}$ for all $\theta \in \paramspace$. Consequently, from Equation \eqref{eqn:gen_objective_stateful}, this implies that for any signalling mechanism $\pi$ and distribution $G$, $|V_{G,h}(\pi) - V_{G,h_\tau}(\pi)| \leq \frac{\epsilon}{4}$. Likewise, this implies that $|V_{G,h_\tau}(\pi_{G,h_\tau}^\ast)-V_{G,h}(\pi_{G,h}^\ast)| \leq \frac{\epsilon}{4}$ since the maximum operator contracts the difference between feasible solutions.

\noindent To show $\epsilon$-optimality of $\hat{\pi}_{\dist_\delta,h_\tau}$, we need to consider four signaling mechanisms. 
\begin{enumerate}
    \item $\pi_{\dist_\delta,\obj_\tau}^\ast = \langle \I_\delta, \{z^\delta_{\nu}\}_{\nu \in \discpspace}\rangle$ is the optimal solution  under the discrete distribution $\dist_\delta$ obtained by solving \textbf{LP($H,\mathbf{y}, \mathbf{c}$)} (Lemma~\ref{lemma:lp_discrete})  where $y_i = \frac{2i-1}{2\tau}$ and $c_{jk} = h(y_k;\nu_j)$.
    \item $\hat{\pi}_{\dist_\delta,\obj_\tau} = \langle \I_\delta , \{\hat{z}_\param\}_{\param \in \paramspace}\rangle$ such that for all $j \in [N], \param \in [\nu_{j-1},\nu_j)$, and $\signal \in \I_\delta$, we have $\hat{z}_\param(\signal) \coloneqq z^\delta_{\nu_j}(\signal)$. That is, the signal distribution in the state $\nu_j$ for the discretized optimal solution is applied to the entire corresponding interval of states in the continuous extension of the mechanism~$\pi_{\dist_\delta,\obj_\tau}^\ast$.
    \item $\pi_{\dist,\obj}^\ast = \langle \I, \{z^\ast_{\param}\}_{\param \in \paramspace} \rangle$ is the unknown true optimal signaling mechanism. This will be useful to bound the quality of $\hat{\pi}_{\dist_\delta,\obj}$.
    \item $\pi_{\dist,\obj}' = \langle \I, \{z'_\param\}_{\param \in \discpspace \subseteq \paramspace} \rangle$ is a discretized modification of $\pi_{\dist,\obj}^\ast$ obtained by averaging the signal distributions over intervals in $\dist_\delta$ as follows:
\end{enumerate}
\begin{align}
\label{eqn:z'}
z'_{\nu_j}(\signal) = \frac{\int_{\frac{j-1}{\delta}}^{\frac{j}{\delta}} z^\ast_\param(\signal) dF(\param)}{p_j}, \quad \forall j \in [N], \signal \in \I. 
\end{align}

\noindent Observe that by the optimality of $\pi^\ast_{\dist_\delta,\obj_\tau}$:
\begin{align*}
    V_{\dist_\delta,\obj}(\pi_{\dist,\obj}')-V_{\dist_\delta,\obj}(\pi^\ast_{\dist_\delta,\obj_\tau}) &\leq \frac{\epsilon}{2} + V_{\dist_\delta,\obj_\tau}(\pi_{\dist,\obj}')-V_{\dist_\delta,\obj_\tau}(\pi^\ast_{\dist_\delta,\obj_\tau}) \\
    &\leq \frac{\epsilon}{2}
\end{align*}
Hence, :
\begin{align}
    V_{\dist,\obj}(\pi_{\dist,\obj}^\ast)&-V_{\dist,\obj}(\hat{\pi}_{\dist_\delta,\obj_\tau})
    =
    V_{\dist,\obj}(\pi_{\dist,\obj}^\ast)-V_{\dist_\delta,\obj}(\pi^\ast_{\dist_\delta,\obj}) + V_{\dist_\delta,\obj}(\pi^\ast_{\dist_\delta,\obj}) -V_{\dist,\obj}(\hat{\pi}_{\dist_\delta,\obj_\tau}) \nonumber \\ 
    &=
    V_{\dist,\obj}(\pi_{\dist,\obj}^\ast)-V_{\dist_\delta,\obj}(\pi_{\dist,\obj}') + V_{\dist_\delta,\obj}(\pi_{\dist,\obj}')-V_{\dist_\delta,\obj}(\pi^\ast_{\dist_\delta,\obj_\tau}) + V_{\dist_\delta,\obj}(\pi^\ast_{\dist_\delta,\obj_\tau}) -V_{\dist,\obj}(\hat{\pi}_{\dist_\delta,\obj_\tau}) \nonumber \\ 
    &\leq
    \underbrace{V_{\dist,\obj}(\pi_{\dist,\obj}^\ast)-V_{\dist_\delta,\obj}(\pi_{\dist,h}')}_{\textbf{(i)}} + \underbrace{V_{\dist_\delta,\obj}(\pi^\ast_{\dist_\delta,\obj_\tau}) -V_{\dist,\obj}(\hat{\pi}_{\dist_\delta,\obj_\tau})}_{\textbf{(ii)}} + \frac{\epsilon}{2} \label{eq:i_ii}
\end{align}
The term~\textbf{(i)} is the loss due to averaging the true optimum's signal distribution across each discretized interval. The term~\textbf{(ii)} represents the loss incurred by applying the discrete optimum's signal generation distribution across intervals. Following Def.~\ref{defn:eps_approx}, bounding each of these terms by $\frac{\epsilon}{4}$ is sufficient to ensure the $\epsilon$-optimality of $\hat{\pi}_{\dist_\delta,h_\tau}$. 

\noindent Observe that in~\textbf{(i)}, both $\pi_{\dist,\obj}^\ast$ and $\pi_{\dist,h}'$ use an identical set of signals $\I$; and in term~\textbf{(ii)} both $\pi^\ast_{\dist_\delta,\obj_\tau}$ and $\hat{\pi}_{\dist_\delta,\obj_\tau}$ use $\I_\delta$. From the construction of $\pi_{\dist,h}'$ (resp.~$\hat{\pi}_{\dist_\delta,\obj_\tau}$) one can conclude that for each signal $\signal \in \I$ (resp. $\signal\in\I_\delta$) the corresponding signal incidence probabilities $q_{\signal}^\ast$ and $q_{\signal}'$ (resp. $\hat{q}_{\signal}$ and $q_{\signal}^\delta$) are equal across $\pi_{\dist,h}^\ast$ and $\pi_{\dist,h}'$ (resp. $\pi_{\dist_\delta,h_\tau}^\ast$ and $\hat{\pi}_{\dist_\delta,h_\tau}$). Moreover, we can conclude from Lemmas~\ref{lemma:posteriors_close_I} and \ref{lemma:posteriors_close_I_delta} that the impact of our chosen discretization scheme on the posterior means induced by $\pi_{\dist,\obj}^\ast$ and $\pi_{\dist,h}'$ (i.e., $\posterior{\signal}^\ast$ and $\posterior{\signal}'$ for $\signal \in \I$) can be controlled; similarly for the posterior means induced by $\pi^\ast_{\dist_\delta,\obj_\tau}$ and $\hat{\pi}_{\dist_\delta,\obj_\tau}$  (i.e., $\posterior{\signal}^\delta$ and $\hat{\posterior{\signal}}$ for $\signal \in \I_\delta$).

\noindent We use these intermediate results to bound~$ V_{\dist,\obj}(\pi_{\dist,\obj}^\ast)-V_{\dist,\obj}(\hat{\pi}_{\dist_\delta,\obj_\tau})$ (see~\eqref{eq:i_ii}) for Lipschitz-continuous models. 

\begin{lemma}
\label{lemma:m_lipschitz}
If $\popdist$ is continuously differentiable with $0 < \frac{d\popdist}{dv} \leq \kappa$, $m(\cdot)$ is $C \kappa$-Lipschitz.
\end{lemma}
\begin{proof}
\noindent Recall that $m(\mu) \coloneqq \inf\{u \geq 0: G^{-1}(u) \geq c_1(u)\mu+c_2(u)\}$ and, by assumption, $G$ is continuously differentiable and monotone so $G^{-1}$ is continuous. Since $c_1, c_2$ are also continuous:
\begin{align*}
    m(\mu) &= \begin{cases}
    u & G^{-1}(u) = c_1(u)\mu + c_2(u), 0 < u < 1 \\
    0 & G^{-1}(0) \geq c_1(0)\mu + c_2(0) \\
    1 & G^{-1}(1) \leq c_1(1)\mu + c_2(1)\\
    \end{cases}
\end{align*}
Suppose by contradiction the claim is false, then there exists $0\leq \mu_1 < \mu_2 \leq M$ with $z_2 \coloneqq m(\mu_2)$, $z_1 \coloneqq m(\mu_1)$ such that $z_2 - z_1 > C \kappa (\mu_2-\mu_1)$. Since $z_1 > z_2$:
\begin{align*}
    \popdist^{-1}(z_1) &\geq c_1(z_1)\mu_1+c_2(z_1)\\
    \popdist^{-1}(z_2) &\leq c_1(z_2)\mu_2+c_2(z_2)
\end{align*}
Therefore:
\begin{align*}
    z_1 &\geq \popdist(c_1(z_1)\mu_1+c_2(z_1))\\
    z_2 &\leq \popdist(c_1(z_2)\mu_2+c_2(z_2))
\end{align*}
Observe that:
\begin{align*}
    \popdist(c_1(z_2)\mu_2+c_2(z_2)) - \popdist(c_1(z_1)\mu_1+c_2(z_1)) &\geq z_2 - z_1 > C \kappa (\mu_2-\mu_1)
\end{align*}
However:
\begin{align*}
    \popdist(c_1(z_2)\mu_2+c_2(z_2)) - \popdist(c_1(z_1)\mu_1+c_2(z_1)) &\leq \kappa(c_1(z_2)\mu_2+c_2(z_2)-c_1(z_1)\mu_1-c_2(z_1)) \\
    &\leq \kappa(c_1(z_2)\mu_2+c_2(z_2)-c_1(z_1)\mu_1-c_2(z_2))\\
    &\leq \kappa(c_1(z_2)\mu_2-c_1(z_1)\mu_1) \\
    &\leq \kappa(c_1(z_2)\mu_2-c_1(z_2)\mu_1) \\
    &\leq C \kappa(\mu_2-\mu_1)
\end{align*}
This is a contradiction, hence $m$ must be $C \kappa$-Lipschitz. 
\end{proof}

\begin{lemma}
    \label{lemma:posteriors_close_I}
    For any $\signal \in \I$, $q_{\signal}^\ast = q_{\signal}'$ and $0 \leq \posterior{\signal}^\ast-\posterior{\signal}' \leq \frac{1}{\delta}$.
\end{lemma}
\begin{proof}
For any $i \in \I$, the discretization scheme and construction of $z_{\nu_j}'$ implies that:
\begin{align*}
q_\signal' &= \sum_{j=1}^N \P_{\param \sim \dist_\delta}[\param = \nu_j] z'_{\nu_j}(\signal) \\
&= \sum_{j=1}^N\int_{\frac{j-1}{\delta}}^{\frac{j}{\delta}} dF(\param) \frac{\int_{\frac{j-1}{\delta}}^{\frac{j}{\delta}} z^\ast_\param(\signal) dF(\param)}{\int_{\frac{j-1}{\delta}}^{\frac{j}{\delta}} dF(\param)}\\
&= \sum_{j=1}^N \int_{\frac{j-1}{\delta}}^{\frac{j}{\delta}} z^\ast_\param(\signal) dF(\param) = \int_{0}^{\maximal} z^\ast_\param(\signal) dF(\param) = q_{\signal}^\ast
\end{align*}
We now bound the difference between $\posterior{i}^\ast$ and $\hat{\posterior{i}}$:
\begin{align*}
\posterior{\signal}^\ast-\posterior{\signal}' &= \frac{\int_{0}^{M} \param z^\ast_{\param}(i) d\dist(\param) }{\int_{0}^{M} z^\ast_{\param}(i) d\dist(\param) }-\frac{\int_{0}^{M} \param z'_{\param}(i) d\dist_\delta(\param) }{\int_{0}^{M} z'_{\param}(i) d\dist_\delta(\param) }\\
&= \frac{\int_{0}^{M} \param z^\ast_{\param}(i) d\dist_\delta(\param) }{q_i^\ast}-\frac{\int_{0}^{M} \param z'_{\param}(i) d\dist(\param) }{q_i'}\\
&=\frac{1}{q_i^\ast}\Big( \int_{0}^{M} \param z^\ast_{\param}(i) d\dist(\param)- \int_{0}^{M} \param z'_{\param}(i) d\dist_\delta(\param)\Big)\\
&= \frac{1}{q_i^\ast}\Big(\sum_{k=1}^{N}\int_{\frac{k-1}{\delta}}^{\frac{k}{\delta}} \param z^\ast_\param(i) dF(\param)-\sum_{k=1}^{N} \nu_k z'_{\nu_k}(i) \int_{\frac{k-1}{\delta}}^{\frac{k}{\delta}}d\dist(\param)\Big)\\
&= \frac{1}{q_i^\ast}\sum_{k=1}^{N} \int_{\frac{k-1}{\delta}}^{\frac{k}{\delta}}(\nu_k-\param)z^\ast_{\theta}(i)d\dist(\param)\\
\intertext{By the discretization scheme, for all $k$, $\theta \in [\frac{k-1}{\delta},\frac{k}{\delta}]$, we know that $0 < \nu_k-\theta \leq \frac{1}{\delta}$:}
0 &\leq \frac{1}{q_i^\ast}\sum_{k=1}^{N} \int_{\frac{k-1}{\delta}}^{\frac{k}{\delta}}(\nu_k-\param)z^\ast_{\theta}(i)d\dist(\param) \leq \frac{1}{\delta q_i^\ast}\sum_{k=1}^{N} \int_{\frac{k-1}{\delta}}^{\frac{k}{\delta}}z^\ast_{\theta}(i)d\dist(\param)\\
0 &\leq \frac{1}{q_i^\ast}\sum_{k=1}^{N} \int_{\frac{k-1}{\delta}}^{\frac{k}{\delta}}(\nu_k-\param)z^\ast_{\theta}(i)d\dist(\param) \leq \frac{1}{\delta q_i^\ast}q_i^\ast\\
\intertext{Therefore:}
0 &\leq \posterior{\signal}^\ast-\posterior{\signal}' \leq \frac{1}{\delta} \quad
\qedsymbol
\end{align*}

\end{proof}
\begin{lemma}
    \label{lemma:posteriors_close_I_delta}
    For any $\signal \in \I_\delta$, $q_{\signal}^\delta = \hat{q}_{\signal}$ and $0 \leq \hat{\posterior{\signal}}- \posterior{\signal}^\delta \leq \frac{1}{\delta}$.
\end{lemma}
\begin{proof}
    The proof is analogous to that of Lemma~\ref{lemma:posteriors_close_I}.
\end{proof}

\noindent The proof proceeds in two parts - bounding terms \textbf{(i)} and \textbf{(ii)} from \eqref{eq:i_ii}. 

\noindent Using 
Lemma \ref{lemma:m_lipschitz} and the conditions of the theorem, we can conclude that planner's objective function $\obj(m(\posterior{\signal}), \trueparam)$ is uniformly $C\kappa\eta_{1}$-Lipschitz in $\posterior{\signal}$.
This is immediate by the conservation of Lipschitz continuity under composition.

We can now analyze (i):
\begin{align*}
       V_{\dist,\obj}(\pi_{\dist,\obj}^\ast)-V_{\dist_\delta,\obj}(\pi_{\dist,h}') &= \mathbb{E}_{\param^\ast \sim \dist, i \sim z^\ast_{\param}}[\obj(m(\posterior{\signal}); \param^\ast)] - \mathbb{E}_{\param' \sim \dist_\delta, \signal \sim 
       \hat{z}_{\param'}}[\obj(m(\posterior{\signal}'); \param')] \\
       &= \sum_{i \in \I}\int_{\param \in \paramspace}\obj(m(\posterior{\signal}); \param) z_{\param}(i) dF(\param) - \sum_{i \in \I}\sum_{\nu \in \paramspace_\delta}\obj(m(\posterior{\signal}'); \nu) p_\nu z_{\nu}(i)  \\
       &= \sum_{i \in \I}\sum_{k=1}^{N}  \Big(\int_{\frac{k-1}{\delta}}^{\frac{k}{\delta}}\obj(m(\posterior{\signal}); \param) z_{\param}(i) d\dist(\param) - \int_{\frac{k-1}{\delta}}^{\frac{k}{\delta}}\obj(m(\posterior{\signal}'); \nu_k) z_{\nu_k}(i) d\dist(\param) \Big) \\
        &\stackrel{\eqref{eqn:z'}}{=} \sum_{i \in \I}\sum_{k=1}^{N}  \Big(\int_{\frac{k-1}{\delta}}^{\frac{k}{\delta}}\big(\obj(m(\posterior{\signal}); \param) -\obj(m(\posterior{\signal}'); \nu_k)\big)z_{\param}(i) d\dist(\param) \Big) \\
        \intertext{By Lemma \ref{lemma:posteriors_close_I}, since $|\nu_k-\param| < \frac{1}{\delta}$ for all $\param \in [\frac{k-1}{\delta},\frac{k}{\delta}]$:}
        &\leq \sum_{i \in \I}\sum_{k=1}^{N}  \Big(\int_{\frac{k-1}{\delta}}^{\frac{k}{\delta}}\big(\obj(m(\posterior{\signal}); \nu_k) -\obj(m(\posterior{\signal}'); \nu_k)+\frac{\eta_{2}}{\delta}\big)z_{\param}(i) d\dist(\param) \Big) \\
         &\leq \sum_{i \in \I}\sum_{k=1}^{N}  \Big(\big(\frac{C\kappa\eta_{1}}{\delta}+\frac{\eta_{2}}{\delta}\big)z_{\param}(i) d\dist(\param) \Big)  = \frac{C\kappa\eta_{1}}{\delta}+\frac{\eta_{2}}{\delta} < \frac{\epsilon}{4}
\end{align*}
We can analogously simplify expression (ii) using the same decomposition. 
\begin{align*}
       V_{\dist_\delta,\obj}(\pi^\ast_{\dist_\delta,\obj_\tau}) -V_{\dist,\obj}(\hat{\pi}_{\dist_\delta,\obj_\tau}) &= \mathbb{E}_{\param^\ast \sim \dist, i \sim z^\ast_{\param}}[\obj(m(\posterior{\signal}); \param^\ast)] - \mathbb{E}_{\param' \sim \dist_\delta, \signal \sim 
       \hat{z}_{\param'}}[\obj(m(\posterior{\signal}'); \param')] \\
       &= \sum_{i \in \I}\int_{\param \in \paramspace}\obj(m(\posterior{\signal}); \param) z_{\param}(i) dF(\param) - \sum_{i \in \I}\sum_{\nu \in \paramspace_\delta}\obj(m(\posterior{\signal}'); \nu) p_\nu z_{\nu}(i)  \\
       &= \sum_{i \in \I}\sum_{k=1}^{N}  \Big(\int_{\frac{k-1}{\delta}}^{\frac{k}{\delta}}\obj(m(\posterior{\signal}); \param) z_{\param}(i) d\dist(\param) - \int_{\frac{k-1}{\delta}}^{\frac{k}{\delta}}\obj(m(\posterior{\signal}'); \nu_k) z_{\nu_k}(i) d\dist(\param) \Big) \\
        &= \sum_{i \in \I}\sum_{k=1}^{N}  \Big(\int_{\frac{k-1}{\delta}}^{\frac{k}{\delta}}\big(\obj(m(\posterior{\signal}); \param) -\obj(m(\posterior{\signal}'); \nu_k)\big)z_{\param}(i) d\dist(\param) \Big) \\
        \intertext{By Lemma \ref{lemma:posteriors_close_I} and Lemma \ref{lemma:m_lipschitz} and since $|\nu_k-\param| < \frac{1}{\delta}$ for all $\param \in [\frac{k-1}{\delta},\frac{k}{\delta}]$:}
        &\leq \sum_{i \in \I}\sum_{k=1}^{N}  \Big(\int_{\frac{k-1}{\delta}}^{\frac{k}{\delta}}\big(\obj(m(\posterior{\signal}); \nu_k) -\obj(m(\posterior{\signal}'); \nu_k)+\frac{\eta_{2}}{\delta}\big)z_{\param}(i) d\dist(\param) \Big) \\
         &\leq \sum_{i \in \I}\sum_{k=1}^{N}  \Big(\big(\frac{C\kappa\eta_{1}}{\delta}+\frac{\eta_{2}}{\delta}\big)z_{\param}(i) d\dist(\param) \Big) = \frac{C\kappa\eta_{1}}{\delta}+\frac{\eta_{2}}{\delta} < \frac{\epsilon}{4}
\end{align*}

\noindent Hence, returning to the general form, $V_{\dist,\obj}(\pi^\ast_{\dist,\obj}) -V_{\dist,\obj}(\hat{\pi}_{\dist_\delta,\obj_\tau}) \leq \epsilon$ and thus $\hat{\pi}_{\dist_\delta,\obj_\tau}$ is $\epsilon$-optimal. \qedsymbol


\subsection{Infectious Cost Model}
\label{appendix:infectious_cost_model} 
\noindent We justify the agent's cost of being of infected \eqref{eqn:infectious_cost} using a simple epidemiological model. We refer to an activity-based model on a complete graph discussed in \cite{hota_generalized_2021,allen_introduction_2008}. Specifically, consider a unit-mass of non-atomic agents over two periods $t\in\{0,1\}$, where each agent begins at $t=0$ in one of the three possible states: susceptible ($S$), asymptomatic ($X$) and symptomatic ($Y$). In this model, both the asymptomatic and symptomatic agents can transmit the disease. Denote the infectious state of each agent $i \in [0,1]$ at time $t$ by $\chi_i(t) \in \{S,X,Y\}$. Assuming that the symptomatic individuals are required to self-isolate, the remaining agents $\mathcal{P} \coloneqq \{i: \chi_i(0) \neq Y\}$ are subject to the decision-making process we consider in Sec. \ref{sec:model}. Since the remaining agents in $\mathcal{P}$ cannot exactly know their existing state, we assume that conditioned on $\chi_i(0) \neq Y$, each agent $i\in\mathcal{P}$ is independently assigned $\chi_i(0) = S$ with probability $p$. Letting the action of agent $i$ in period $t=0$ be $a_i \in \{\ws,\wrr\}$, the mass of agents working in person is $m = \int_{i \in \mathcal{P}} \mathbb{I}\{a_i=\ws \}$ and the mass of asymptomatic agents working in-person is $m_X = \int_{i \in \mathcal{P}} \mathbb{I}\{a_i=\ws \}\mathbb{I}\{\chi_i(0)= X\}$. Given the risk of contracting disease from a single contact -- what we refer to as the risk parameter $\param$ -- any initially susceptible agent $i$ with $\chi_i(0) = S$ and $a_i = \ws$ transitions to being infected in period $t=1$ (i.e., $\chi_i(1) \in \{X,Y\}$) with probability $\param m_X$ for small $\param$. Only susceptible agents will pay an incremental infectious cost as the remaining agents were already infected. Specifically, agents $i$ incur a cost $\gamma$ if and only if $\chi_i(0) = S$ and $\chi_i(1) \in \{X,Y\}$. Hence, if $a_i = \wrr$, agent $i$ has no contact with other individuals and hence the infection cost she fares in expectation is 0. On the other hand, if $a_i = \ws$, then agent $i$ pays an expected cost $\beta(\param,m)$:
\begin{align*}
  \beta(\param,m) &=  \mathbb{E}[\gamma \indicator{\chi_i(0) =S \wedge \chi_i(1) \in \{X,Y\}}] \\
  &= \gamma \mathbb{P}[\chi_i(0) =S]\mathbb{P}[\chi_i(1) \in \{X,Y\}\}|\chi_i(0) =S] \\
  &= \gamma (1-p) \param \mathbb{E}[m_X] \\
  &= \gamma p(1-p) \param m
\end{align*}

\noindent This model suggests that an agent's expected cost of infection has a linear dependence on $\param$ and the mass of agents working in-person $m$ (which is $1-\fin$ in our setting). This is consistent with the functional form of infectious costs in \eqref{eqn:infectious_cost}; in particular, when choosing $c_1(u) = 1-u$ and $c_2(u) = 0$. More generally, as the network structure underlying the infection dynamics becomes specialized or other diseases become intermingled, the associated $c_i$ may be better estimated through other functions that satisfy the assumptions we make on these terms. For the purpose of numerical experiments in Sec. \ref{subsec:numerical_state_indpt}, we consider $\beta(\param, \fin)~=~\param (1-\fin)$.

\subsection{Optimal signaling against No-Information and Full-Information Benchmarks}
\label{subsec:numerical_state_indpt}
\begin{figure}[h!]
    \centering
    \includegraphics[width=6cm]{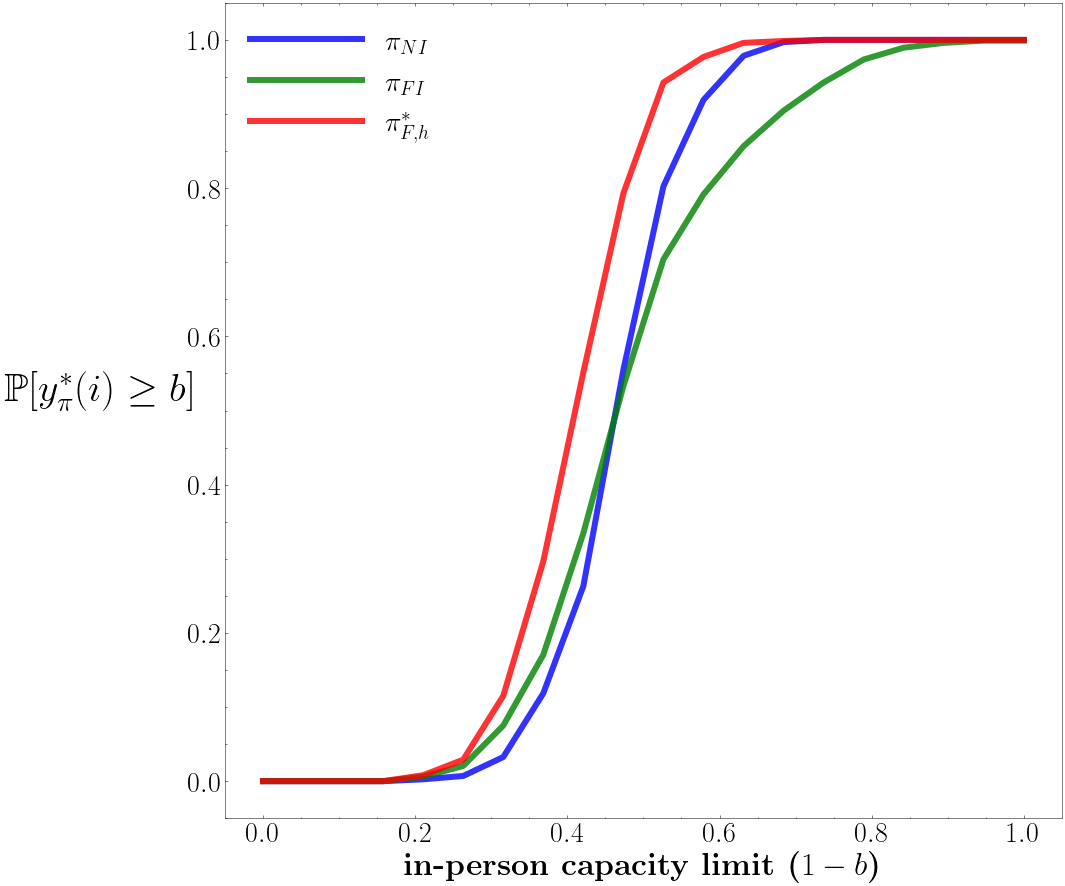}
    \caption{Capacity compliance across non-informative, fully-informative and optimal signaling.}
    \label{fig:stateless_carying_capacity}
\end{figure}
We present a numerical example to compare the optimal signaling mechanism in Theorem~\ref{thm:r123} with no-information and full-information benchmarks. Consider the planner's set-based preference~$\goal = \{\fin : \fin \geq b\}$; i.e., the planner prefers that the mass of in-person agents $(1-\fin)$ is below a threshold capacity limit~$(1-b)$ of the workplace facility. Let $\dist$ be uniform on $[5, 20]$ ($\priormean=12.5$) and $\popdist$ uniform on~$[0,10]$. The planner seeks to maximize the probability that in-person mass complies with the capacity limit. 

\noindent We vary $b$ from $0$ to $1$ in increments of~$0.05$ and use 10,000 scenario runs for each~$b$ to obtain the outcomes of optimal signaling mechanism and the two benchmarks. The results are averaged across scenarios and shown in Fig.~\ref{fig:stateless_carying_capacity}. The optimal signaling mechanism indeed provides higher compliance relative to the two benchmarks, but the improvement reduces as $b \nearrow 1$ and as $b \searrow 1$. As $b \nearrow 1$, the intersection between the outcomes achievable in equilibrium and the ones preferred by the planner progressively reduces to $0$, hence the effectiveness of signaling in influencing agents decreases. On the other hand, as~$b \searrow 1$ the set of acceptable outcomes grows to encompass all outcomes, and the optimal signaling as well as no- and full-information benchmarks approach full compliance. Thus, the value of optimal signaling decreases as the planner's set of acceptable outcomes grows. 

\subsection{State-dependent non-MPS mechanisms at optimality}
\label{ex:statefulpool}
Consider $\obj(\fin,\trueparam) = |\frac{\trueparam}{3} - y|$ with $m(\mu) = \mu$. Let $\dist$ be uniform over $[0,1]$ ($\priormean = 0.5$). Observe $h$ is uniformly $2$-Lipschitz in both $\fin$ and $\trueparam$ over $\fin,\trueparam \in [0,1]$. This example clarifies that modifying the probability of generating signal $\signal\in\I$ when the state is $\param$ (i.e. $z_{\param}(\signal)$) is not straightforward -- any modification in the mechanism through the parameter $z_{\param}(\signal)$ to change $\posterior{\signal}$ so that $m( \posterior{\signal}) = \final(\signal)$ moves away from $\param$ can result in $\final(\signal)$ moving toward $\param'$ for some other state $\param'$ that also maps to $\signal$ (i.e. $z_{\param}(\signal)>0$). 

\noindent Thus, the planner may elect to generate signals in a way that forgoes some utility in some realizations of the state $\trueparam$ to generate posterior means that are more preferred for other realizations of the state $\trueparam$. This complicates the structure of the optimal signaling mechanisms. In particular, a monotone partitional structure no longer holds and the optimal solution requires pooling of states from disconnected regions of the state-space. 

\noindent To show this, consider any mechanism $\pi$ satisfying MPS. Then there exists $\{t_k\}_{j=0}^K$ with $K \geq 0$, $t_0 = 0$ and $t_K = 1$ where $\pi$ generates signal $j \in [K]$ exactly when $\trueparam \in [t_{j-1}, t_{j}]$. Consequently, this mechanism induces a posterior mean~$\posterior{j}=\frac{t_{j-1}+t_{j}}{2}$ for signal~$j$. Hence the objective attained by the planner using this mechanism is:
\begin{align}
V_{F,\obj}(\pi) &= \mathbb{E}_{\trueparam \sim \dist, i \sim z_{\trueparam}}\big[\obj(\final(i); \trueparam)\big] \\
    &= \E\big[\E[\lvert \frac{\trueparam}{3}-\posterior{j}\rvert \big| t_{j-1}\leq \trueparam < t_j]\big]
\end{align}

\noindent For any $j$ and for all $\trueparam \in [t_{j-1}, t_{j}]$ observe that $\posterior{j} = \frac{t_{j-1}+t_{j}}{2} \geq \frac{t_{j}}{2} > \frac{t_{j}}{3} \geq \frac{\trueparam}{3}$, which implies that $\lvert \frac{\trueparam}{3}-\posterior{j}\rvert = \posterior{j} - \frac{\trueparam}{3}$. Hence by tower rule and mean-preservation of the posteriors $\posterior{j}$:
\begin{align}
V_{F,\obj}(\pi) &= \E\big[\E[\posterior{j}-\frac{\trueparam}{3}\big| t_{j-1}\leq \trueparam < t_j]\big] \\
&= \E[\posterior{j}]-\E[\frac{\trueparam}{3}]\\
&= \priormean -  \frac{1}{3} \priormean = \frac{1}{3} 
\end{align}
Hence, we have shown that every MPS mechanism achieves the same objective value of $\frac{1}{3}$.

\noindent Now consider the mechanism: $\pi^\mathrm{pool} = \langle \I, \{z_\param\}_{\param\in\paramspace}\rangle$ where $\I = \{\mathbf{1},\mathbf{2},\mathbf{3}\}$ and $z_{\theta}(s)$ is as follows:
\begin{align*}
z_\theta(\cdot) = \begin{cases}
& \mathbf{1} \text{ w.p. } 1 \text{ if }\param\in\mathcal{S}_1\coloneqq[0,0.25]\\
& \mathbf{2} \text{ w.p. } 1 \text{ if }\param\in\mathcal{S}_2\coloneqq[0.25,0.35]\cup[0.95,1]\\
& \mathbf{3} \text{ w.p. } 1 \text{ if }\param\in\mathcal{S}_3\coloneqq[0,1]\setminus\{\mathcal{S}_1\cup\mathcal{S}_2\}
\end{cases}
\end{align*}
The resulting posterior means are $\posterior{\mathbf{1}} = 0.125, \posterior{\mathbf{2}} = 0.525$ and $\posterior{\mathbf{3}} = 0.65$. Enumerating the objective, $\pi^\mathrm{pool}$ achieves an objective value of $\approx 0.4854$, exceeding the performance of any MPS-based mechanism. 

\subsection{Comparison with \cite{de_vericourt_informing_2021}}
\label{subsec:dv_comp}
We first perform a direct comparison with the model studied in~\cite{de_vericourt_informing_2021}, which corresponds to the following planner utility in our setting: 
\begin{align*}
\obj_{\text{ref}(\lambda)}(\fin;\trueparam) = \lambda \mathbb{E}_{v \sim G}[v\mathbb{I}\{v\geq G^{-1}(\fin)\}]- (1-\lambda) \trueparam (1-\fin)^2.
\end{align*}
Under the model they present, $\popdist \sim Unif[0,6]$ is the distribution of agent's value of in-person work, and  continuous-valued state $\trueparam\sim Unif[0,10]$ (in contrast to the binary-valued state in~\cite{de_vericourt_informing_2021}. That is, planner's expected utility is the expected gain of all the agents who choose in-person work (i.e, $\mathbb{E}_{v \sim G}[v\mathbb{I}\{v\geq G^{-1}(\fin)\}]$) net the {total} expected disutility incurred by these agents in facing risk of disease transmission ($\trueparam (1-\fin)^2$). The weight $\lambda\in[0,1]$ captures the tradeoff between the two terms. 

\noindent We do a direct comparison of our computational approach for approximating $\pi_{\dist,\obj_{\text{ref}(\lambda)}}^\ast$ against their closed-form optimal solution. Particularly, we replicate the preference model $\obj_{\text{ref}(\lambda)}$ of ~\cite{de_vericourt_informing_2021} and consider a binary model of uncertainty $\trueparam \sim \bar{\dist}$ that takes value $\trueparam = 0$ with probability $\frac{1}{2}$ and $\trueparam = 10$ with probability $\frac{1}{2}$. We apply our numerical approach to find approximate solutions for various levels of discretization when $\lambda\in[0,0.25,0.5,0.75,1]$. The results of ~\cite{de_vericourt_informing_2021} provide a closed form representation of $\pi_{\bar{\dist},\obj_{\text{ref}(\lambda)}}^\ast$ and show that full information $\pi_{FI}$ is optimal except when $\lambda=1$. In this setting, observe our algorithm need not discretize $\bar{\dist}$ as it is already a discrete distribution. In Fig. \ref{fig:discrete_error}, we plot the error in the objective value achieved between using the computed signaling mechanism and the true optimal signaling mechanism as we vary the discretization $\tau$ used in approximating $\obj_{\text{ref}(\lambda)}$. As shown, our approach recovers a signaling mechanism with hardly any discretization when the true optimum is full information disclosure. When $\lambda = 1$, the convergence is slower -- but still faster than the rate provided in Theorem \ref{thm:stateful}. When the optimal signaling mechanism is no longer fully-informative, the errors in the solution to the chosen linear program solution $z_{ji}^\ast$ accumulate more heavily in the derived signaling mechanism as constructed in Lemma \ref{lemma:lp_discrete}. This occurs because the solutions $z_{ji}^\ast$ for partial-information disclosure mechanisms are no longer as sparse as fully-informative or non-informative disclosure rules. In general, however, our results do recover those of ~\cite{de_vericourt_informing_2021} and our rate of convergence beats the theoretical bound we provide.

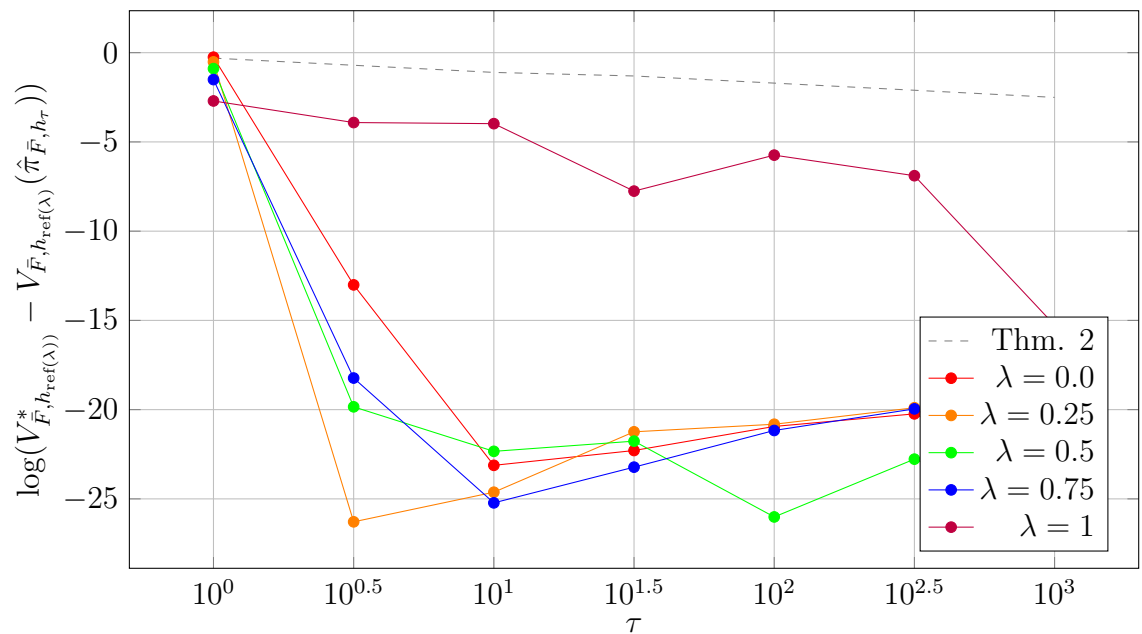
\begin{figure}[h!]
\centering
\begin{tikzpicture}
\begin{axis}[
symbolic x coords={1,3,10,32,100,316,1000},
xticklabels={$10^{0}$,$10^{0.5}$,$10^{1}$,$10^{1.5}$,$10^{2}$,$10^{2.5}$,$10^{3}$},
xtick=data,
height=9cm,
width=15cm,
grid=major,
xlabel={$\tau$},
ylabel={$\log(V_{\bar{F},h_{\text{ref}(\lambda))}}^\ast - V_{\bar{F},h_{\text{ref}(\lambda)}}(\hat{\pi}_{\bar{F},h_{\tau}}))$},
legend style={
cells={anchor=east},
legend pos=south east,
}
]
\addplot[-,dashed,gray] coordinates {
(1,-0.30541906202998762)  (3,-0.70541906202998762)  (10,-1.10541906202998762)  (32,-1.30541906202998762)  (100,-1.70541906202998762)  (316,-2.10541906202998762)
(1000,-2.50541906202998762)};
\addplot[red,mark=*] coordinates {
(1,-0.25235032842418237)  
(3,-13.010250810089534)  
(10,-23.120460721604918)  
(32,-22.290475254182006)  
(100,-20.94179643776927)  
(316,-20.236889404001346)  
(1000,-23.878350656130195)  
};
\addplot[orange,mark=*] coordinates {
(1,-0.5161178830572486)  
(3,-26.285730048361632)  
(10,-24.62650327462636)  
(32,-21.24231941351772)  
(100,-20.81839271960233)  
(316,-19.89459928481808)  
(1000,-25.544142354256344)  
};
\addplot[green,mark=*] coordinates {
(1,-0.8935811616169201)  
(3,-19.842766871673753)  
(10,-22.336555778521685)  
(32,-21.765509735176526)  
(100,-26.011860576428106)  
(316,-22.774048286700893)  
(1000,-25.862761398540382)  
};
\addplot[blue,mark=*] coordinates {
(1,-1.5071409042585173)  
(3,-18.230108877330288)  
(10,-25.223175349592598)  
(32,-23.222195124494082)  
(100,-21.167450848201103)  
(316,-19.95764362880023)  
(1000,-23.78327723444541)  
};
\addplot[purple,mark=*] coordinates {
(1,-2.7052113077937587)  
(3,-3.915120628419362)  
(10,-3.9769864826846133)  
(32,-7.756332294300116)  
(100,-5.740690183119474)  
(316,-6.893089519121701)  
(1000,-15.266485227863647)  
};

\legend{Thm. \ref{thm:stateful},$\lambda = 0.0$,$\lambda = 0.25$,$\lambda = 0.5$,$\lambda = 0.75$,$\lambda = 1$}
\end{axis}
\end{tikzpicture}
\caption{Error of computed $\epsilon$-optimal solution $\hat{\pi}_{\bar{F},h_{\tau}}$ as discretization ($\tau$) increases and the weights $\lambda$ are varied.}
\label{fig:discrete_error}
\end{figure}


\end{document}